\documentclass[journal,final,onecolumn]{IEEEtran}
\usepackage{fixltx2e}
\usepackage{cite}
\usepackage{mathrsfs}
\usepackage{color}
\usepackage{float}
\usepackage[caption=false]{subfig}
\usepackage[table, dvipsnames]{xcolor}
\usepackage{array,booktabs}
\ifCLASSINFOpdf
   \usepackage[pdftex]{graphicx}
   \graphicspath{{Figs/}}
   \DeclareGraphicsExtensions{.pdf,.jpeg,.png}
\else
\fi
\usepackage[cmex10]{amsmath}
\usepackage{amsmath}
\usepackage{amssymb}
\usepackage{amsthm}
\usepackage{amsfonts}
\usepackage{bm}
\usepackage{xfrac}
\usepackage{empheq}
\usepackage[normalem]{ulem} 
\usepackage{soul} 
\usepackage{mathtools}

\usepackage{etoolbox}
\makeatletter
\patchcmd{\@makecaption}
{\scshape}
{}
{}
{}
\makeatletter
\patchcmd{\@makecaption}
{\\}
{.\ }
{}
{}
\makeatother

\newtheorem{theorem}{Theorem}

\newtheorem{example}{Example}
\newtheorem{proposition}{Proposition}
\newtheorem{lemma}{Lemma}

\newtheorem{corollary}{Corollary}
\newtheorem{remark}{Remark}
\theoremstyle{definition}
\newtheorem{definition}{Definition}
\usepackage{xcolor}
\usepackage{color}
\graphicspath{{figs/}}
\begin{document}
	\title{On the Privacy-Utility Trade-off With and Without Direct Access to the Private Data}
\vspace{-5mm}
\author{
	\IEEEauthorblockN{Amirreza Zamani, ~\IEEEmembership{Member,~IEEE,} Tobias J. Oechtering,~\IEEEmembership{Senior Member,~IEEE,} Mikael Skoglund,~\IEEEmembership{Fellow,~IEEE} \vspace*{0.5em}
	}
	\thanks{This work was funded in
		part by the Swedish research council under contract 2019-03606. This work was presented in part at the 2022 IEEE International Symposium on Information Theory and the 2022 IEEE Information Theory Workshop. A. Zamani, M. Skoglund and T. J. Oechtering are with the Division of Information Science and Engineering, School of Electrical Engineering and
		Computer Science, KTH Royal Institute of Technology, 100 44 Stockholm,
		Sweden (e-mail: amizam@kth.se; oech@kth.se; skoglund@kth.se).}}
\maketitle

\begin{abstract}
	We study an information theoretic privacy mechanism design problem for two scenarios where the private data is either observable or hidden. In each scenario, we first consider bounded mutual information as privacy leakage criterion, then we use two different per-letter privacy constraints. In the first scenario, 
	an agent observes useful data $Y$ that is correlated with private data $X$, and wishes to disclose the useful information to a user.
	A privacy mechanism is designed to generate disclosed data $U$ which maximizes the revealed information about $Y$ while satisfying a bounded privacy leakage constraint. In the second scenario, the agent has additionally access to the private data. To this end, we first extend the Functional Representation Lemma and Strong Functional Representation Lemma by relaxing the independence condition and thereby allowing a certain leakage to find lower bounds for the second scenario with different privacy leakage constraints. Furthermore, upper and lower bounds are derived in the first scenario considering different privacy constraints. In particular, for the case where no leakage is allowed, our upper and lower bounds improve previous bounds. Moreover, considering bounded mutual information as privacy constraint we show that if the common information and mutual information between $X$ and $Y$ are equal, then the attained upper bound in the second scenario is tight. Finally, the privacy-utility trade-off with prioritized private data is studied where part of $X$, i.e., $X_1$, is more private than the remaining part, i.e., $X_2$, and we provide lower and upper bounds.  
\end{abstract}

\section{Introduction}
 The privacy mechanism design problem from an information theory perspective is recently receiving increased attention and related results can be found in \cite{ makhdoumi, issa, issa2, Calmon2,yamamoto, sankar, dwork11, dwork222, cuff2016differential, oech, asoodeh1, Total, borz, gun,khodam,Khodam22,7888175, wang, kostala, dwork1, calmon4, issajoon , zamani2022bounds, zamani2022, zamani2022multi, courtade, sankar2, deniz4, asoodeh3, Calmon1,  nekouei2}.

In more detail, in \cite{makhdoumi}, the concept of a privacy funnel is introduced, where the privacy utility trade-off has been studied considering a distortion measure for utility and the log-loss as privacy measure. The concept of maximal leakage has been introduced in \cite{issa} and used in \cite{issa2} for the Shannon cipher system.
Furthermore, some bounds on the privacy-utility trade-off are derived. 
Fundamental limits of the privacy utility trade-off measuring the leakage using estimation-theoretic guarantees are studied in \cite{Calmon2}.
A related secure source coding problem is studied in \cite{yamamoto}.

In both \cite{yamamoto} and \cite{sankar}, the privacy-utility trade-offs considering expected distortion and equivocation as a measures of utility and privacy are studied.
 The concept of differential privacy is introduced in \cite{dwork11} and it has been used in \cite{dwork222} to answer queries in a privacy-preserving approach using minimizing the chance of identifying the membership in a statistical database. The concept of mutual information as differential privacy is introduced in \cite{cuff2016differential}. 
 In \cite{oech}, the hypothesis test performance of an adversary is used to measure privacy leakage.
\begin{figure}[]
	\centering
	\includegraphics[scale = .15]{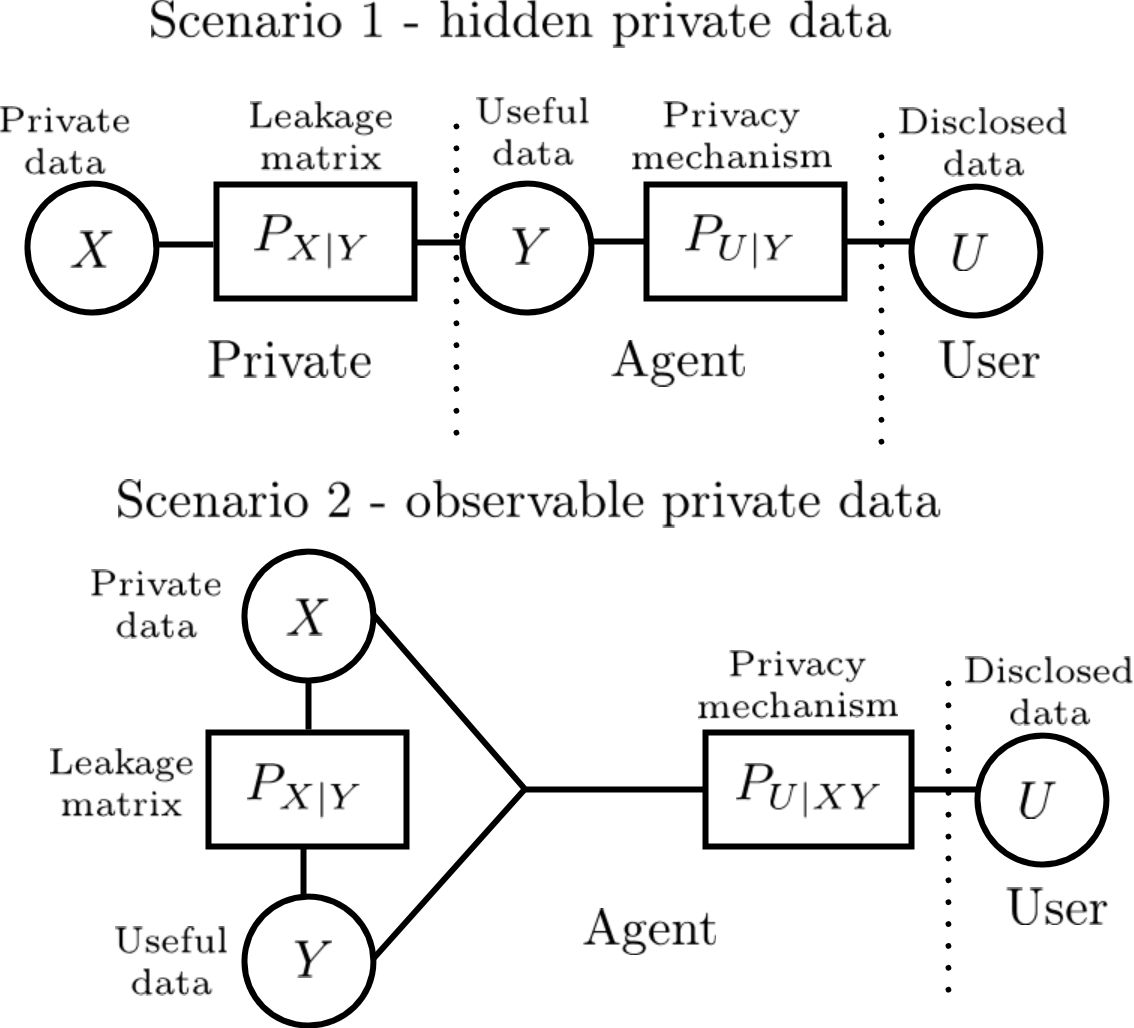}
	\caption{In the first scenario the agent has only access to $Y$ and in the second scenario the agent has additionally access to $X$.}
	\label{ISITsys}
\end{figure}
 In \cite{asoodeh1}, maximal correlation either mutual information is used for measuring the privacy and properties of rate-privacy functions are studied.
 In \cite{Total}, average total variation is used as a privacy measure and a $\chi^2$-privacy criterion is considered in \cite{Calmon2}, where an upper bound and a lower bound on the privacy-utility trade-off have been derived.
The problem of privacy-utility trade-off considering mutual information both as measures of utility and privacy given the Markov chain $X-Y-U$ is studied in \cite{borz}. Under the perfect privacy assumption it is shown that the privacy mechanism design problem can be reduced to a linear program. This has been extended in \cite{gun} considering the privacy utility trade-off with a rate constraint on the disclosed data.
Moreover, in \cite{borz}, it has been shown that information can be only revealed if the kernel (leakage matrix) between useful data and private data is not invertible. In \cite{khodam}, we generalize \cite{borz} by relaxing the perfect privacy assumption allowing some small bounded leakage. More specifically, we design privacy mechanisms with a per-letter privacy criterion considering an invertible kernel where a small leakage is allowed. We generalized this result to a non-invertible leakage matrix in \cite{Khodam22}.
In both \cite{7888175} and \cite{ wang}, the optimal privacy-utility trade-offs have been studied considering two scenarios where the private data is either observable or hidden. Sufficient conditions for equality of the optimal trade-offs in the considered two scenarios have been derived where the utility is measured by a distortion metric.  
In \cite{kostala}, by using the Functional Representation Lemma bounds on privacy-utility trade-off for the two scenarios are derived. These results are derived under the perfect secrecy assumption, i.e., no leakages are allowed. The bounds are tight when the private data is a deterministic function of the useful data.

In this paper, random variable (RV) $Y$ denotes the useful data and is correlated with the private data denoted by RV $X$. Furthermore, RV $U$ describes the disclosed data. Two scenarios are considered in this work, where in both scenarios, an agent wants to disclose the useful information to a user as shown in Fig.~\ref{ISITsys}. In the first scenario, the agent observes $Y$ and has no direct access to $X$, i.e., the private data is hidden. The goal is to design $U$ based on $Y$ that reveals as much information as possible about $Y$ and satisfies a bounded privacy criterion. In the second scenario, the agent has access to both $X$ and $Y$ and can design $U$ based on $(X,Y)$ to release as much information as possible about $Y$ while satisfying the bounded leakage constraint.
In both scenarios we consider different privacy constraints. 
Our results in this work can be divided into three main parts as follows:\\
	\textbf{Part I (\emph{Privacy-utility trade-off with non-zero leakage}):} In the first part of the paper, our problem is closely related to \cite{kostala}, where the problem of \emph{secrecy by design} is studied. We generalize the privacy problems considered in \cite{kostala} by relaxing the perfect privacy constraint and allowing some leakage. More specifically, we consider bounded mutual information, i.e., $I(U;X)\leq \epsilon$ for privacy leakage constraint. To this end, we extend the Functional Representation Lemma and the Strong Functional Representation Lemma, introduced in \cite{kosnane} by relaxing the independence condition to derive lower bounds for the second scenario. We show that if the common information and mutual information between $X$ and $Y$ are equal, then the maximum utility in two scenarios and the attained upper bound in the second scenario are equal. Furthermore, in the special case of perfect privacy we find a new upper bound for the perfect privacy function by using the \emph{excess functional information} introduced in \cite{kosnane}.
	We show that this new bound generalizes the bound in \cite{kostala}. Moreover, we show that the bound is tight when $|\mathcal{Y}|=2$. Finally, we compare our new lower and upper bounds with the bounds found in \cite{kostala} when the leakage is zero. The conference version regarding this part can be found in \cite{zamani2022bounds}.\\
    \textbf{Part II (\emph{Privacy-utility trade-off with non-zero leakage and per-letter privacy constraints}):} In the second part, for each scenario we use two different per-letter privacy constraints instead of the bounded mutual information constraint. As argued in \cite{Khodam22}, it can be more desirable to protect the private data individually and not merely on average. 
    We first find similar results as the extended versions of the Functional Representation Lemma and the Strong Functional Representation Lemma found in the previous part considering the per-letter privacy constraint rather than bounded mutual information. Using these results we find a lower bound for the privacy-utility trade-off in the second scenario. Furthermore, we provide bounds for three other problems and study a special case where $X$ is a deterministic function of $Y$. We show that the obtained upper and lower bounds in the first scenario are asymptotically optimal when $X$ is a deterministic function of $Y$. In \cite{Khodam22}, one of the problems considered in this part has been studied. It has been shown that by using methods from Euclidean information geometry as used in \cite{Shashi, huang}, we can simplify the design problem in the high privacy regime and the main problem can be solved approximately by a linear program. In this work, we provide upper bounds on the error of the approximation considered in \cite{Khodam22}. Finally we compare the attained bounds in a numerical example. The conference version related to this part can be found in \cite{zamani2022bounds}.\\
    \textbf{Part III (\emph{Privacy-utility trade-off with non-zero leakage and prioritized private data}):}
    Finally, we consider the problem in the second scenario where the private data is divided into two parts, i.e., $X=(X_1,X_2)$. In this part we use bounded mutual information as privacy constraint. We assume that the first part is more private than the second part, i.e., the privacy leakage of $X_1$ is less than or equal to the privacy leakage of $X_2$. Furthermore, we assume that the total leakage between $(X_1,X_2)$ and $U$ is bounded by $\epsilon$ and we derive upper and lower bounds. Similar to the previous parts we use the extended versions of Functional Representation Lemma and the Strong Functional Representation Lemma to find lower bounds. \\
    Our contribution can be summarized as follows:\\
    \textbf{(i)} We extend the Functional Representation Lemma and the Strong Functional Representation Lemma by a randomized response output that allows some controlled leakage. Various extended versions are introduced using different leakage measures.  \\ 
    \textbf{(ii)} We formulate and study various privacy mechanism design problems through the lens of information theory with controlled leakage, demonstrating the use of the extended versions of Functional Representation Lemma and the Strong Functional Representation Lemma.\\
      \textbf{(iii)} We provide discussion and comparison of the obtained results with each other and the literature.\\
\textbf{Notation:}
Given two jointly random variables $X$ and $Y$, the entropy, conditional entropy and mutual information between $X$ and $Y$ are given by
$H(Y) = \mathbb{E}(\log(\frac{1}{P_Y(y)}))$,
$H(Y|X) = \mathbb{E}(\log(\frac{1}{P_{Y|X}(y|x)}))$,
and $I(X;Y)=H(Y)-H(Y|X)$.
$X$ and $Y$ are independent if and only if 
$
I(X;Y)=0.
$
Furthermore, the Markov chain $X-Y-U$ holds if and only if
$
I(X;U|Y)=0.
$
For the binary entropy $h(\cdot)$ we have
$
h(p)=-\left(p\log(p)+(1-p)\log(1-p)\right).
$
In this work, let matrix $P_{XY}$ defined on $\mathbb{R}^{|\mathcal{X}|\times|\mathcal{Y}|}$ denote the joint distribution of discrete random variables $X$ and $Y$ defined on finite alphabets $\cal{X}$ and $\cal{Y}$. We represent  
marginal distributions of $X$ and $Y$ by vectors $P_X$ and $P_Y$ defined on $\mathbb{R}^{|\mathcal{X}|}$ and $\mathbb{R}^{|\mathcal{Y}|}$ given by the row and column sums of $P_{XY}$. We represent the leakage matrix $P_{X|Y}$ by a matrix defined on $\mathbb{R}^{|\mathcal{X}|\times|\cal{Y}|}$ with elements $P_{X|Y}(x|y)$ for all $x$ and $y$. Furthermore, for given $u\in \mathcal{U}$, $P_{X,U}(\cdot,u)$ and $P_{X|U}(\cdot|u)$ defined on $\mathbb{R}^{|\mathcal{X}|}$ are distribution vectors with elements $P_{X,U}(x,u)$ and $P_{X|U}(x|u)$ for all $x\in\cal X$ and $u\in \cal U$. 
The relation between $U$ and $Y$ is described by the kernel $P_{U|Y}$ defined on $\mathbb{R}^{|\mathcal{U}|\times|\mathcal{Y}|}$, furthermore, the relation between $U$ and the pair $(Y,X)$ is described by the kernel $P_{U|Y,X}$ defined on $\mathbb{R}^{|\mathcal{U}|\times|\mathcal{Y}|\times|\mathcal{X}|}$.
\section{system model and Problem Formulation} \label{sec:system}
In this work we assume that each element in vectors $P_X$ and $P_Y$ is non-zero. 
In the second part of the results, which corresponds to \emph{privacy-utility trade-off with non-zero leakage and per-letter privacy constraints}, we assume that for the discrete random variables $X$ and $Y$ defined on finite alphabets $\cal{X}$ and $\cal{Y}$ we have that $|\mathcal{X}|<|\mathcal{Y}|$.
 Furthermore, we assume that the leakage matrix $P_{X|Y}$ is of full rank. In the remaining parts of the results we consider arbitrary correlated discrete random variables $X$ and $Y$ as private data and useful data.\\
In the following we introduce the main problems in three different parts. In each part, we first define the problems considered in this paper, then we motivate them and study the properties of the measures for utility and privacy leakage and compare them with previous works. 
\subsection{Privacy-utility trade-off with non-zero leakage}
In this part, for both design problems we use mutual information as utility and leakage measures. The privacy mechanism design problems for the two scenarios can be stated as follows
\begin{align}
g_{\epsilon}(P_{XY})&=\sup_{\begin{array}{c} 
	\substack{P_{U|Y}:X-Y-U\\ \ I(U;X)\leq\epsilon,}
	\end{array}}I(Y;U),\label{main2}\\
h_{\epsilon}(P_{XY})&=\sup_{\begin{array}{c} 
	\substack{P_{U|Y,X}: I(U;X)\leq\epsilon,}
	\end{array}}I(Y;U).\label{main1}
\end{align} 
The function $h_{\epsilon}(P_{XY})$ is used when the privacy mechanism has access to both the private data and the useful data. The function $g_{\epsilon}(P_{XY})$ is used when the privacy mechanism has only access to the useful data. Clearly, the relation between $h_{\epsilon}(P_{XY})$ and $g_{\epsilon}(P_{XY})$ can be stated as follows
\begin{align*}
g_{\epsilon}(P_{XY})\leq h_{\epsilon}(P_{XY}).
\end{align*}
In the following we study the case where $0\leq\epsilon< I(X;Y)$, otherwise the optimal solution of $h_{\epsilon}(P_{XY})$ or $g_{\epsilon}(P_{XY})$ is $H(Y)$ achieved by $U=Y$. 
\begin{remark}
	\normalfont
	For $\epsilon=0$, \eqref{main2} leads to the perfect privacy problem studied in \cite{borz}. It has been shown that for a non-invertible leakage matrix $P_{X|Y}$, $g_0(P_{XY})$ can be obtained by a linear program.
	Furthermore, for $\epsilon=0$, \eqref{main1} leads to the secret-dependent perfect privacy function $h_0(P_{XY})$, studied in \cite{kostala}, where upper and lower bounds on $h_0(P_{XY})$ have been derived. The bounds are tight when $X$ is deterministic function of $Y$.
\end{remark}
\subsection{Privacy-utility trade-off with non-zero leakage and per-letter privacy constraints}
The privacy mechanism design problems for the two scenarios can be stated as follows 
\begin{align}
g_{\epsilon}^{w\ell}(P_{XY})&=\sup_{\begin{array}{c} 
	\substack{P_{U|Y}:X-Y-U\\ \ d(P_{X,U}(\cdot,u),P_XP_{U}(u))\leq\epsilon,\ \forall u}
	\end{array}}I(Y;U),\label{main22}\\
h_{\epsilon}^{w\ell}(P_{XY})&=\sup_{\begin{array}{c} 
	\substack{P_{U|Y,X}: d(P_{X,U}(\cdot,u),P_XP_{U}(u))\leq\epsilon,\ \forall u}
	\end{array}}I(Y;U),\label{main11}\\
g_{\epsilon}^{\ell}(P_{XY})&=\sup_{\begin{array}{c} 
	\substack{P_{U|Y}:X-Y-U\\ \ d(P_{X|U}(\cdot|u),P_X)\leq\epsilon,\ \forall u}
	\end{array}}I(Y;U),\label{main222}\\
h_{\epsilon}^{\ell}(P_{XY})&=\sup_{\begin{array}{c} 
	\substack{P_{U|Y,X}: d(P_{X|U}(\cdot|u),P_X)\leq\epsilon,\ \forall u}
	\end{array}}I(Y;U),\label{main12}
\end{align} 
where $d(P,Q)$ corresponds to the total variation distance between two distributions $P$ and $Q$, i.e., $d(P,Q)=\sum_x |P(x)-Q(x)|$.
The functions $h_{\epsilon}^{w\ell}(P_{XY})$ and $h_{\epsilon}^{\ell}(P_{XY})$ are used when the privacy mechanism has access to both the private data and the useful data. The functions $g_{\epsilon}^{w\ell}(P_{XY})$ and $g_{\epsilon}^{\ell}(P_{XY})$ are used when the privacy mechanism has only access to the useful data. In this work, the privacy constraints used in \eqref{main22} and \eqref{main222}, i.e., $d(P_{X,U}(\cdot,u),P_XP_U(u))\leq\epsilon,\ \forall u,$ and $d(P_{X|U}(\cdot|u),P_X)\leq\epsilon,\ \forall u,$ are called the \emph{weighted strong $\ell_1$-privacy criterion} and the \emph{strong $\ell_1$-privacy criterion}. We refer to them as strong since they are per-letter privacy constraints, i.e., they must hold for every $u\in\cal U$. The difference between the two privacy constraints in this work is the weight $P_U(u)$, therefore, we refer to $d(P_{X,U}(\cdot,u),P_XP_U(u))\leq\epsilon,\ \forall u,$ as weighted.
We later show that the weight $P_U(u)$ enables us to use extended versions of the Functional Representation Lemma and Strong Functional Representation Lemma to find lower bounds considering the second scenario. 
\begin{remark}
	\normalfont
	We have used the leakage constraint $d(P_{X|U}(\cdot|u),P_X)\leq\epsilon,\ \forall u$ in \cite{Khodam22}, where we called it the \emph{strong $\ell_1$-privacy criterion}. 
\end{remark}  
\begin{remark}
	\normalfont
	For $\epsilon=0$, both \eqref{main22} and \eqref{main222} lead to the perfect privacy problem studied in \cite{borz}. It has been shown that for a non-invertible leakage matrix $P_{X|Y}$, $g_0(P_{XY})$ can be obtained by a linear program.
\end{remark}
\begin{remark}
	\normalfont
	For $\epsilon=0$, both \eqref{main11} and \eqref{main12} lead to the secret-dependent perfect privacy function $h_0(P_{XY})$, studied in \cite{kostala}, where upper and lower bounds on $h_0(P_{XY})$ have been derived. In \cite{zamani2022bounds}, we have strengthened these bounds.
\end{remark}
\begin{remark}
	\normalfont
	The privacy problem defined in \eqref{main222} has been studied in \cite{Khodam22} where we provide a lower bound on $g_{\epsilon}^{\ell}(P_{XY})$ using the information geometry concepts. Furthermore, we have shown that without loss of optimality it is sufficient to assume $|\mathcal{U}|\leq |\mathcal{Y}|$ so that it is ensured that the supremum can be achieved. 
\end{remark}
Intuitively, for small $\epsilon$, both privacy constraints mean that $X$ and $U$ are almost independent. As we discussed in \cite{Khodam22}, closeness of $P_{X|U}(\cdot|u)$ and $P_X$ allows us to approximate $g_{\epsilon}^{\ell}(P_{XY})$ with a series expansion and find a lower bound. In this work we show that by using a similar methodology, we can approximate $g_{\epsilon}^{w\ell}(P_{XY})$ exploiting the closeness of $P_{X,U}(\cdot,u)$ and $P_XP_U(u)$. This provides us a lower bound for $g_{\epsilon}^{w\ell}(P_{XY})$. Next, we study some properties of the weighted strong $\ell_1$-privacy criterion and the strong $\ell_1$-privacy criterion. To this end recall that the \emph{linkage inequality} is the property that if $\cal L$ measures the privacy leakage between two random variables and the Markov chain $X-Y-U$ holds then we have $\mathcal{L}(X;U)\leq\mathcal{L}(Y;U)$. Since the weighted strong $\ell_1$-privacy criterion and the strong $\ell_1$-privacy criterion are per letter constraints we define $\mathcal{L}^1(X;U=u)\triangleq \left\lVert P_{X|U}(\cdot|u)-P_X \right\rVert_1$, $\mathcal{L}^1(Y;U=u)\triangleq \left\lVert P_{Y|U}(\cdot|u)-P_Y \right\rVert_1$, $\mathcal{L}^2(X;U=u)\triangleq \left\lVert P_{X,U}(\cdot,u)-P_XP_U(u) \right\rVert_1$, $\mathcal{L}^2(Y;U=u)\triangleq \left\lVert P_{Y,U}(\cdot,u)-P_YP_U(u) \right\rVert_1$.
\begin{proposition}\label{gooz}
	The weighted strong $\ell_1$-privacy criterion and the strong $\ell_1$-privacy criterion satisfy the linkage inequality. Thus, for each $u\in\mathcal{U}$ we have $\mathcal{L}^1(X;U=u)\leq\mathcal{L}^1(Y;U=u)$ and $\mathcal{L}^2(X;U=u)\leq\mathcal{L}^2(Y;U=u)$.	
\end{proposition}
\begin{proof}
	The proof is provided in Appendix A. 
	 
\end{proof}
As discussed in \cite[page 4]{Total}, one benefit of the linkage inequality is to keep the privacy in layers of private information which is discussed in the following. Assume that the Markov chain $X-Y-U$ holds and the distribution of $X$ is not known. If we can find $\tilde{X}$ such that $X-\tilde{X}-Y-U$ holds and the distribution of $\tilde{X}$ is known then by the linkage inequality we can conclude $\mathcal{L}(X;U=u)\leq \mathcal{L}(\tilde{X};U=u)$. In other words, if the framework is designed for $\tilde{X}$, then a privacy constraint on $\tilde{X}$ leads to the constraint on $X$, i.e., provides an upper bound for any pre-processed RV $X$. To have the Markov chain $X-\tilde{X}-Y-U$ consider the scenario where $\tilde{X}$ is the private data and $X$ is a function of private data which is not known. For instance let $\tilde{X}=(X_1,X_2,X_3)$ and $X=X_1$. Thus, the mechanism that is designed based on $\tilde{X}-Y-U$ preserves the leakage constraint on $X$ and $U$. As pointed out in \cite[Remark 2]{Total}, among all the $L^p$-norms ($p\geq 1$), only the $\ell_1$ norm satisfies the linkage inequality. Next, given a leakage measure $\mathcal{L}$ and let the Markov chain $X-Y-U$ hold, if we have  $\mathcal{L}(X;U)\leq \mathcal{L}(X;Y)$, then we say that the \emph{post processing inequality} holds. In this work we use $\mathcal{L}^1(X;U)= \sum_u P_U(u)\mathcal{L}^1(X;U=u)$, $\mathcal{L}^2(X;U)= \sum_u \mathcal{L}^2(X;U=u)$ and $\mathcal{L}^1(Y;U)=\sum_u P_U(u)\mathcal{L}^1(Y;U=u)$, $\mathcal{L}^2(Y;U)=\sum_u \mathcal{L}^2(Y;U=u)$.
\begin{proposition}
	The average of the weighted strong $\ell_1$-privacy criterion and the strong $\ell_1$-privacy criterion with weights equal one and $P_U(u)$, respectively, satisfy the post-processing inequality, i.e., we have $\mathcal{L}^1(X;U)\leq\mathcal{L}^1(Y;U)$ and $\mathcal{L}^2(X;U)\leq\mathcal{L}^2(Y;U)$.	  
\end{proposition}
\begin{proof}
	The proof is the same as that of \cite[Theorem 3]{Total} which is based on the convexity of the $\ell_1$-norm.
\end{proof}
\begin{proposition}
	The weighted strong $\ell_1$-privacy criterion and the strong $\ell_1$-privacy criterion result in bounded inference threat that is modeled in \cite{Calmon1}.
\end{proposition}
\begin{proof}
	The weighted strong $\ell_1$-privacy criterion and the strong $\ell_1$-privacy criterion lead to a bounded on average constraint $\sum_u P_U(u)\left\lVert P_{X|U=u}\!-\!P_X \right\rVert_1=2TV(X;U)\leq \epsilon$, where $TV(.|.)$ corresponds to the total variation. Thus, using \cite[Theorem 4]{Total}, we conclude that adversarial inference performance is bounded. 
\end{proof}
Another property of the $\ell_1$ distance is the relation between the $\ell_1$-norm and probability of error in a hypothesis test. As argued in \cite[Remark~6.5]{polyanskiy2014lecture}, for the binary hypothesis test with $H_0:X\sim P$ and $H_1:X\sim Q$,
the expression $1-TV(P, Q)$ is the sum of false alarm and missed detection probabilities. Thus, we have $TV(P,Q)=1-2P_e$, where $P_e$ is the error probability (the probability that we can not decide the right distribution for $X$ with equal prior probabilities for $H_0$ and $H_1$). For instance, consider the scenario where we want to decide whether $X$ and $U$ are independent or correlated. To this end, let $P=P_{X,U}$, $Q=P_{X}P_{U}$, $H_0:X,U\sim P$ and $H_1:X,U\sim Q$. We have
\begin{align*}
TV(P_{X,U};P_{X}P_{U})&=\frac{1}{2}\sum_u P_U(u)\left\lVert P_{X|U=u}-P_X \right\rVert_1\!\\&\leq \frac{1}{2}\epsilon.
\end{align*}
Thus, by increasing the leakage, which means that $TV(P_{X,U};P_{X}P_{U})$ increases, then the error of probability decreases.  \\
Finally, if we use $\ell_1$ distance as privacy leakage, after approximating $g_{\epsilon}^{w\ell}(P_{XY})$ and $g_{\epsilon}^{\ell}(P_{XY})$, we face a linear program problems in the end, which are much easier to handle.
\subsection{Privacy-utility trade-off with non-zero leakage and prioritized private data}
In this part, we assume that the private data $X$ is divided into two parts $X_1$ and $X_2$, where the first part is more private than the other part, i.e., the privacy leakage of $X_1$ is less than or equal to the privacy leakage of $X_2$. We use mutual information for measuring both privacy leakage and utility and we only consider the second scenario where the privacy mechanism has access to both $X$ and $Y$. Hence, the problem can be stated as follows
   \begin{align}
   h_{\epsilon}^{p}(P_{X_1X_2Y})&=\sup_{\begin{array}{c} 
   	\substack{P_{U|YX_1X_2}: I(U;X_1,X_2)\leq\epsilon,\\ I(U;X_1)\leq I(U;X_2)}
   	\end{array}}I(Y;U).\label{main111}
   \end{align}
The constraint $I(U;X_1,X_2)\leq\epsilon$ ensures that the total leakage is bounded by $\epsilon$ and the constraint $I(U;X_1)\leq I(U;X_2)$ corresponds to the priority of $X_1$. In practice, we usually have different levels of privacy leakage for the private data and in this work we consider two levels.
\begin{remark}
	\normalfont
	For $\epsilon=0$, \eqref{main111} leads to the secret-dependent perfect privacy function $h_0(P_{XY})$.
\end{remark}
\section{Relation between the problems}\label{rel}
In this section we study the relation between the privacy measures that are used. In the following, we first present the relation between the weighted strong $\ell_1$-privacy criterion and bounded mutual information.
\begin{proposition}\label{trash1}
	For any $\epsilon\geq 0$ and pair $(X,U)$ we have
	\begin{align}\label{tr1}
	I(X;U)\leq \epsilon\Rightarrow d(P_{X,U}(\cdot,u),P_XP_U(u))\leq\sqrt{2\epsilon},\ \forall u.
	\end{align}
	\begin{proof}
		We have
		\begin{align*}
		\epsilon\geq I(U;X)&=\sum_u P_U(u)D(P_{X|U}(\cdot|u),P_X)\\ &\stackrel{(a)}{\geq} \sum_u \frac{P_U(u)}{2}\left( d(P_{X|U}(\cdot|u),P_X)\right)^2\\ &\geq \frac{P_U(u)^2}{2}\left( d(P_{X|U}(\cdot|u),P_X)\right)^2\\ &= \frac{\left( d(P_{X,U}(\cdot,u),P_XP_U(u))\right)^2}{2},
		\end{align*}
		where (a) follows by the Pinsker's inequality \cite{verdu}.
	\end{proof}
\end{proposition}  
\begin{corollary}\label{trash2}
	By using \eqref{tr1} we have
	\begin{align*}
	h_{\bar{\epsilon}}^{w\ell}(P_{XY})\leq h_{\epsilon}(P_{XY}),\\
		g_{\bar{\epsilon}}^{w\ell}(P_{XY})\leq g_{\epsilon}(P_{XY}),
	\end{align*}
	where $\bar{\epsilon}=\sqrt{2\epsilon}$.
\end{corollary}
Next, we present the relation between the strong $\ell_1$-privacy criterion and bounded mutual information.
\begin{proposition}\label{trash3}
	For any $\epsilon\geq 0$ and pair $(X,U)$ we have
	\begin{align}\label{tr2}
	d(P_{X|U}(\cdot|u),P_X)\leq\epsilon,\ \forall u \Rightarrow  I(X;U)\leq \frac{\epsilon^2}{\min P_X}.
	\end{align}
	\begin{proof}
		We have
		\begin{align*}
		I(X;U)& \stackrel{(a)}{\leq}\sum_u P_U(u)\frac{\left(d(P_{X|U}(\cdot|u),P_X)\right)^2}{\min P_X}  \\& = \frac{\epsilon^2}{\min P_X},
		\end{align*}
		where (a) follows by the reverse Pinsker's inequality \cite{verdu}.
	\end{proof}
\end{proposition} 
\begin{corollary}\label{trash4}
	By using \eqref{tr2} we have
	\begin{align*}
	h_{\epsilon'}(P_{XY})\leq h_{\epsilon}^{\ell}(P_{XY}),\\
	g_{\epsilon'}(P_{XY})\leq g_{\epsilon}^{\ell}(P_{XY}),
	\end{align*}
	where $\epsilon'=\frac{\epsilon^2}{\min P_X}$.
\end{corollary}
Using Corollary~\ref{trash2} and Corollary~\ref{trash4} we have
\begin{align*}
h_{\bar{\epsilon}}^{w\ell}(P_{XY})\leq h_{\epsilon}(P_{XY})\leq h_{\tilde{\epsilon}}^{\ell}(P_{XY}),\\
g_{\bar{\epsilon}}^{w\ell}(P_{XY})\leq g_{\epsilon}(P_{XY})\leq g_{\tilde{\epsilon}}^{\ell}(P_{XY}),
\end{align*}
where $\tilde{\epsilon}=\sqrt{\epsilon\min P_X}$.
\section{Main results}\label{result}
In this part, we provide lower and upper bounds for the privacy problems defined in \eqref{main2}, \eqref{main1}, \eqref{main22}, \eqref{main11}, \eqref{main222}, \eqref{main12} and \eqref{main111}. We study the tightness of the bounds in special cases and compare them in examples. In more detail, in the first part of the results, which corresponds to \emph{privacy-utility trade-off with non-zero leakage}, we show that the upper bound on $h_{\epsilon}(P_{XY})$ is achieved when the common information and mutual information between $X$ and $Y$ are equal. We provide necessary and sufficient conditions for the achievability of the obtained upper bound in general. Moreover, in cases where no leakage is allowed, i.e., $\epsilon=0$, we provide new bounds that generalize the previous bounds. In the second part of the results in this section corresponding to \emph{privacy-utility trade-off with non-zero leakage and per-letter privacy criterions} we use concepts from information geometry to find lower bounds on $g_{\epsilon}^{w\ell}(P_{XY})$ and $g_{\epsilon}^{\ell}(P_{XY})$. In the remaining parts of the main results, we provide lower and upper bounds for $h_{\epsilon}^{p}(P_{X_1X_2Y})$ and study them for special cases.
\subsection{Privacy-utility trade-off with non-zero leakage}
   In this section, we first recall the Functional Representation Lemma (FRL) \cite[Lemma~1]{kostala} and Strong Functional Representation Lemma (SFRL) \cite[Theorem~1]{kosnane} for discrete $X$ and $Y$. Then we extend them for correlated random variables $X$ and $U$, i.e., $0\leq I(U;X)=\epsilon$. We refer to them as Extended Functional Representation Lemma (EFRL) and Extended Strong Functional Representation Lemma (ESFRL). We show that the extended lemmas, i.e., EFRL and ESFRL, enable us to find lower bounds on $h_{\epsilon}(P_{XY})$.
   \begin{lemma}\label{lemma1} (Functional Representation Lemma \cite[Lemma~1]{kostala}):
   	For any pair of RVs $(X,Y)$ distributed according to $P_{XY}$ supported on alphabets $\mathcal{X}$ and $\mathcal{Y}$ where $|\mathcal{X}|$ is finite and $|\mathcal{Y}|$ is finite or countably infinite, there exists a RV $U$ defined on $\mathcal{U}$ such that $X$ and $U$ are independent, i.e., we have
   	\begin{align}\label{c1}
   	I(U;X)=0,
   	\end{align}
   	$Y$ is a deterministic function of $(U,X)$, i.e., we have
   	\begin{align}
   	H(Y|U,X)=0,\label{c2}
   	\end{align}
   	and 
   	\begin{align}
   	|\mathcal{U}|\leq |\mathcal{X}|(|\mathcal{Y}|-1)+1.\label{c3}
   	\end{align}
   \end{lemma}
 \begin{lemma}\label{lemma2} (Strong Functional Representation Lemma \cite[Theorem~1]{kosnane}):
	For any pair of RVs $(X,Y)$ distributed according to $P_{XY}$ supported on alphabets $\mathcal{X}$ and $\mathcal{Y}$ where $|\mathcal{X}|$ is finite and $|\mathcal{Y}|$ is finite or countably infinite with $I(X,Y)< \infty$, there exists a RV $U$ defined on $\mathcal{U}$ such that $X$ and $U$ are independent, i.e., we have
	\begin{align*}
	I(U;X)=0,
	\end{align*}
	$Y$ is a deterministic function of $(U,X)$, i.e., we have 
	\begin{align*}
	H(Y|U,X)=0,
	\end{align*}
	$I(X;U|Y)$ can be upper bounded as follows
	\begin{align*}
	I(X;U|Y)\leq \log(I(X;Y)+1)+4,
	\end{align*}
	and 
	$
	|\mathcal{U}|\leq |\mathcal{X}|(|\mathcal{Y}|-1)+2.
	$
\end{lemma}
\begin{remark}
	By checking the proof in \cite[Th.~1]{kosnane}, the term $e^{-1}\log(e)+2+\log(I(X;Y)+e^{-1}\log(e)+2)$ can be used instead of $\log(I(X;Y)+1)+4$. 
\end{remark}
\begin{remark}
	Idea of extending Functional Representation Lemma and Strong Functional Representation Lemma is basically adding a randomized response argument to the random variable $U$ found by Lemma~\ref{lemma1} and Lemma~\ref{lemma2}. The idea is simple, when it can be connected with an old principle it gets creditable. 
\end{remark}
\begin{lemma}\label{lemma3} (Extended Functional Representation Lemma):
	For any $0\leq\epsilon< I(X;Y)$ and pair of RVs $(X,Y)$ distributed according to $P_{XY}$ supported on alphabets $\mathcal{X}$ and $\mathcal{Y}$ where $|\mathcal{X}|$ is finite and $|\mathcal{Y}|$ is finite or countably infinite, there exists a RV $U$ defined on $\mathcal{U}$ such that the leakage between $X$ and $U$ is equal to $\epsilon$, i.e., we have
	\begin{align*}
	I(U;X)= \epsilon,
	\end{align*}
	$Y$ is a deterministic function of $(U,X)$, i.e., we have  
	\begin{align*}
	H(Y|U,X)=0,
	\end{align*}
	and 
	$
	|\mathcal{U}|\leq \left[|\mathcal{X}|(|\mathcal{Y}|-1)+1\right]\left[|\mathcal{X}|+1\right].
	$
\end{lemma}
\begin{proof}
	The proof is provided in Appendix B.
\end{proof}
\begin{lemma}\label{lemma4} (Extended Strong Functional Representation Lemma):
	For any $0\leq\epsilon< I(X;Y)$ and pair of RVs $(X,Y)$ distributed according to $P_{XY}$ supported on alphabets $\mathcal{X}$ and $\mathcal{Y}$ where $|\mathcal{X}|$ is finite and $|\mathcal{Y}|$ is finite or countably infinite with $I(X,Y)< \infty$, there exists a RV $U$ defined on $\mathcal{U}$ such that the leakage between $X$ and $U$ is equal to $\epsilon$, i.e., we have
	\begin{align*}
	I(U;X)= \epsilon,
	\end{align*}
	$Y$ is a deterministic function of $(U,X)$, i.e., we have 
	\begin{align*}
	H(Y|U,X)=0,
	\end{align*}
	$I(X;U|Y)$ can be  upper bounded as follows 
	\begin{align*}
	I(X;U|Y)\leq \alpha H(X|Y)+(1-\alpha)\left[ \log(I(X;Y)+1)+4\right],
	\end{align*}
	and 
	$
	|\mathcal{U}|\leq \left[|\mathcal{X}|(|\mathcal{Y}|-1)+2\right]\left[|\mathcal{X}|+1\right],
	$
	where $\alpha =\frac{\epsilon}{H(X)}$.
\end{lemma}
\begin{proof}
	The proof is provided in Appendix~B.
\end{proof}
In Lemma~\ref{haa}, which is proved in Appendix B, we show that there exists a RV $U$ that satisfies \eqref{c1}, \eqref{c2} and has bounded entropy. The lemma is a generalization of
\cite[Lemma~2]{kostala} for dependent $X$ and $U$. \\
Before stating the next theorem we derive an expression for $I(Y;U)$. We have
\begin{align}
I(Y;U)&=I(X,Y;U)-I(X;U|Y),\nonumber\\&=I(X;U)+I(Y;U|X)-I(X;U|Y),\nonumber\\&=I(X;U)+H(Y|X)-H(Y|U,X)-I(X;U|Y).\label{key}
\end{align}
As argued in \cite{kostala}, \eqref{key} is an important observation to find lower and upper bounds for $h_{\epsilon}(P_{XY})$ and $g_{\epsilon}(P_{XY})$.\\
Next theorem characterizes the constraints that the utility in the second scenario is larger than $\epsilon$, i.e., $h_{\epsilon}(P_{XY})>\epsilon$. Furthermore, we provide a necessary condition under which this result holds. For the sufficiency part, we use the construction of RV $U$ which is stated in Lemma~\ref{lemma3} and Lemma~\ref{lemma4}.  
\begin{theorem}\label{an}
	For any $0\leq \epsilon< I(X;Y)$ and pair of RVs $(X,Y)$ distributed according to $P_{XY}$ supported on alphabets $\mathcal{X}$ and $\mathcal{Y}$, if $h_{\epsilon}(P_{XY})>\epsilon$ then we have
	\begin{align*}
	H(Y|X)>0.
	\end{align*}
	Furthermore, if $H(Y|X)-\alpha H(X|Y)-(1-\alpha)\min\{H(X|Y),\left[ \log(I(X;Y)+1)+4\right]\}>0$, then
	\begin{align*}
	h_{\epsilon}(P_{XY})>\epsilon,
	\end{align*}
	where $\alpha=\frac{\epsilon}{H(X)}$.
\end{theorem}
\begin{proof}
	For proving the first part let $h_{\epsilon}(P_{XY})>\epsilon$. Using \eqref{key} we have
	\begin{align*}
	\epsilon&< h_{\epsilon}(P_{XY})\leq H(Y|X)+\sup_{U:I(X;U)\leq\epsilon} I(X;U)\\&=H(Y|X)+\epsilon
	\Rightarrow 0<H(Y|X).
	\end{align*}
	For the second part first assume that 
	$H(X|Y)\leq \log(I(X;Y)+1)+4$, which results in the assumption
	$H(Y|X)-H(X|Y)>0$. In this case, let $U$ be produced by EFRL. Thus, using the construction of $U$ as in Lemma~\ref{lemma3} we have $I(X,U)=\epsilon$ and $H(Y|X,U)=0$. Then by using \eqref{key} we obtain
	\begin{align*}
	h_{\epsilon}(P_{XY})&\geq \epsilon\!+\!H(Y|X)-H(X|Y)+H(X|Y,U)\\&\geq \epsilon+H(Y|X)-H(X|Y)\\&>\epsilon.
	\end{align*}
	Now assume that $\log(I(X;Y)+1)+4 \leq H(X|Y)$, which results in the assumption \\ $H(Y|X)-\alpha H(X|Y)-(1-\alpha)\left[ \log(I(X;Y)+1)+4\right]>0$. In this case, let $U$ be produced by the ESFRL. Thus, using the construction of $U$ as in Lemma~\ref{lemma4} we have $I(X,U)=\epsilon$, $H(Y|X,U)=0$ and $I(U;X|Y)\leq \alpha H(X|Y)+(1-\alpha)\left[ \log(I(X;Y)+1)+4\right]$. Then by using \eqref{key} we obtain
	\begin{align*}
	h_{\epsilon}(P_{XY})&\geq \epsilon\!+\!H(Y|X)\!-\!I(X;U|Y)\\&\geq \epsilon\!+\!H(Y|X)-\alpha H(X|Y)-(1-\alpha)\left[ \log(I(X;Y)+1)+4\right]\\&>\epsilon.
	\end{align*}
\end{proof}
In the next theorem we present lower bounds on $h_{\epsilon}(P_{XY})$ and find the conditions under which the bounds are tight. The following theorem is a generalization of \cite[Th.~6]{kostala} for correlated $X$ and $U$, i.e., $I(X;U)\leq \epsilon$.
\begin{theorem}\label{th.1}
	For any $0\leq \epsilon< I(X;Y)$ and pair of RVs $(X,Y)$ distributed according to $P_{XY}$ supported on alphabets $\mathcal{X}$ and $\mathcal{Y}$ we have
	\begin{align}\label{th2}
	h_{\epsilon}(P_{XY})\geq \max\{L_1^{\epsilon},L_2^{\epsilon},L_3^{\epsilon}\},
	\end{align}
	where
	\begin{align*}
	L_1^{\epsilon} &= H(Y|X)-H(X|Y)+\epsilon=H(Y)-H(X)+\epsilon,\\
	L_2^{\epsilon} &= H(Y|X)-\alpha H(X|Y)+\epsilon-(1-\alpha)\left( \log(I(X;Y)+1)+4 \right),\\
	L_3^{\epsilon} &= \epsilon\frac{H(Y)}{I(X;Y)}+g_0(P_{XY})\left(1-\frac{\epsilon}{I(X;Y)}\right),
	\end{align*}
	and $\alpha=\frac{\epsilon}{H(X)}$.
	The lower bound in \eqref{th2} is tight if $H(X|Y)=0$, i.e., $X$ is a deterministic function of $Y$. Furthermore, if the lower bound $L_1$ is tight then we have $H(X|Y)=0$. 
\end{theorem}
\begin{proof}
	The proof is provided in Appendix B. Similar to Theorem~\ref{an}, the lower bounds $L_1^{\epsilon}$ and $L_2^{\epsilon}$ are derived by using Lemma~\ref{lemma3} and Lemma~\ref{lemma4}.
\end{proof}
\begin{corollary}\label{sos}
	Using \eqref{key} utility achieved by FRL is $H(Y|X)-H(X|Y)$ which is less than or equal to utility achieved by SFRL, i.e., $H(Y|X)-H(X|Y)+\epsilon=L_1^{\epsilon}$. Furthermore, utility achieved by SFRL is $H(Y|X)-\left( \log(I(X;Y)+1)+4 \right)$ which is less than or equal to utility attained by ESFRL, i.e., $H(Y|X)+\epsilon-\alpha H(X|Y)-(1-\alpha)\left( \log(I(X;Y)+1)+4 \right)=L_2^{\epsilon}$, since we have
	\begin{align*}
	L_2^{\epsilon}-\left(H(Y|X)-\left( \log(I(X;Y)+1)+4 \right)\right)&=\epsilon+\frac{\epsilon}{H(X)}\left( \log(I(X;Y)+1)+4\right)-\frac{\epsilon}{H(X)}H(X|Y)\\ &\geq 0.
	\end{align*}
	The latter holds since $H(X|Y)\leq H(X)$. Equality holds if and only if $\epsilon=0$. Hence, for non-zero leakage EFRL and ESFRL strictly improve the bounds attained by FRL and SFRL.
\end{corollary}
In next corollary we let $\epsilon=0$ and derive lower bounds on $h_0(P_{XY})$.
\begin{corollary}\label{kooni}
	Let $\epsilon=0$. Then, for any pair of RVs $(X,Y)$ distributed according to $P_{XY}$ supported on alphabets $\mathcal{X}$ and $\mathcal{Y}$ we have
	\begin{align*}
	h_{0}(P_{XY})\geq \max\{L^0_1,L^0_2\},
	\end{align*}
	where 
	\begin{align*}
	L^0_1 &= H(Y|X)-H(X|Y)=H(Y)-H(X),\\
	L^0_2 &= H(Y|X) -\left( \log(I(X;Y)+1)+4 \right).\\
	\end{align*}
\end{corollary}
Note that the lower bound $L^0_1$ has been derived in \cite[Th.~6]{kostala}, while the lower bound $L^0_2$ is new. Hence, the lower bound derived in Corollary~\ref{kooni} generalizes the bound found in \cite[Th.~6]{kostala}. \\
In the next two examples we compare the bounds $L_1^{\epsilon}$, $L_2^{\epsilon}$ and $L_3^{\epsilon}$ in special cases where $I(X;Y)=0$ and $H(X|Y)=0$.  
\begin{example}
	Let $X$ and $Y$ be independent. Then, we have
	\begin{align*}
	L_1^{\epsilon}&=H(Y)-H(X)+\epsilon,\\
	L_2^{\epsilon}&=H(Y)-\frac{\epsilon}{H(X)}H(X)+\epsilon-4(1-\frac{\epsilon}{H(X)}),\\
	&=H(Y)-4(1-\frac{\epsilon}{H(X)}).
	\end{align*}
	Thus,
	\begin{align*}
	L_2^{\epsilon}-L_1^{\epsilon}&=H(X)-4+\epsilon(\frac{4}{H(X)}-1),\\
	&=(H(X)-4)(1-\frac{\epsilon}{H(X)}).
	\end{align*}
	Consequently, for independent $X$ and $Y$ if $H(X)>4$, then $L_2^{\epsilon}>L_1^{\epsilon}$, i.e., the second lower bound is dominant and $h_{\epsilon}(P_{X}P_{Y})\geq L_2^{\epsilon}$. 
\end{example}
\begin{example}
	Let $X$ be a deterministic function of $Y$. As we have shown in Theorem~\ref{th.1}, if $H(X|Y)=0$, then
	\begin{align*}
	L_1^{\epsilon}&=L_3^{\epsilon}=H(Y|X)+\epsilon\\&\geq H(Y|X)+\epsilon-(1-\frac{\epsilon}{H(X)})(\log(H(X)+1)+4)\\&=L_2^{\epsilon}.
	\end{align*}
	Therefore, $L_1^{\epsilon}$ and $L_3^{\epsilon}$ become dominants.
\end{example}
	In Lemma~\ref{koonimooni} which is provided in Appendix B, we find a lower bound for $\sup_{U} H(U)$ where $U$ satisfies the leakage constraint $I(X;U)\leq\epsilon$, the bounded cardinality stated in Lemma~\ref{lemma3} and $ H(Y|U,X)=0$.\\
	In the next result, using \eqref{key} we derive an upper bound on $h_{\epsilon}(P_{XY})$.
\begin{lemma}\label{goh}
	For any $0\leq\epsilon< I(X;Y)$ and pair of RVs $(X,Y)$ distributed according to $P_{XY}$ supported on alphabets $\mathcal{X}$ and $\mathcal{Y}$ we have
	\begin{align*}
	g_{\epsilon}(P_{XY})\leq h_{\epsilon}(P_{XY})\leq H(Y|X)+\epsilon.
	\end{align*}
\end{lemma}
\begin{proof}
	By using \eqref{key} we have
	\begin{align*}
	h_{\epsilon}(P_{XY})\leq H(Y|X)+\sup I(U;X)\leq H(Y|X)+\epsilon. 
	\end{align*}
\end{proof}
\begin{corollary}\label{seff}
	If $X$ is a deterministic function of $Y$, then by using Theorem~\ref{th.1} and Lemma~\ref{goh} we have
	\begin{align*}
	g_{\epsilon}(P_{XY})=h_{\epsilon}(P_{XY})=H(Y|X)+\epsilon,
	\end{align*}
	since in this case the Markov chain $X-Y-U$ holds.
\end{corollary}
In the next result we find a larger set of distributions $P_{XY}$ compared to Corollary~\ref{seff} for which we have $g_{\epsilon}(P_{XY})=h_{\epsilon}(P_{XY})$, where common information corresponds to the Wyner \cite{wyner} or G{\'a}cs-K{\"o}rner \cite{gacs1973common} notions of common information. One advantage of having $g_{\epsilon}(P_{XY})=h_{\epsilon}(P_{XY})$ is discussed after Theorem~\ref{ziba}, where we show that under the assumption of equality between common information and mutual information the upper bound in Lemma~\ref{goh} is tight.
\begin{proposition}\label{sef}
	If the common information and the mutual information between $X$ and $Y$ are equal, then we have
	\begin{align*}
	g_{\epsilon}(P_{XY})=h_{\epsilon}(P_{XY}).
	\end{align*}
\end{proposition}
\begin{proof}
	The proof follows similar arguments as the proof of \cite[Th.~2]{7888175}. Let $U^*$ be an optimizer of $h_{\epsilon}(P_{XY})$, then by using the proof of \cite[Th.~2]{7888175} we can construct $U'$ satisfying the Markov chain $X-Y-U'$, $I(U^*;Y)=I(U';Y)$ and $I(U';X)\leq I(U^*;X)$ which completes the proof.
\end{proof}
Clearly, if $X$ is a deterministic function of $Y$, then the common information and mutual information between $X$ and $Y$ are equal. Thus, the constraint in Proposition~\ref{sef} contains a larger set of joint distributions $P_{XY}$ compared to the constraint used in Corollary~\ref{seff}.\\ In the next lemma we provide an important property of an optimizer of $h_{\epsilon}(P_{XY})$ which is used to derive equivalencies in Theorem~\ref{ziba}. 
\begin{lemma}\label{kir}
	Let $\bar{U}$ be an optimizer of $h_{\epsilon}(P_{XY})$. We have
	\begin{align*}
	H(Y|X,\bar{U})=0.
	\end{align*}
\end{lemma}
\begin{proof}
	The detailed proof is provided in Appendix B and is similar to the proof of \cite[Lemma~5]{kostala}. The proof is by contradiction and we show that if for an optimizer of $h_{\epsilon}(P_{XY})$, denoted by $\bar{U}$, we have $H(Y|X,\bar{U})>0$, then we can build $U$ such that it satisfies $I(U;X)\leq \epsilon$ and achieves strictly greater utility than $\bar{U}$, which contradicts the assumption.  
\end{proof}
In the next theorem we generalize the equivalent statements in \cite[Th.~7]{kostala} for bounded leakage between $X$ and $U$, i.e., $I(X;U)\leq \epsilon$.
\begin{theorem}\label{ziba}
	For any $\epsilon<I(X;Y)$, we have the following equivalencies
	\begin{itemize}
		\item [i.] $g_{\epsilon}(P_{XY})=H(Y|X)+\epsilon$,
		\item [ii.] $g_{\epsilon}(P_{XY})=h_{\epsilon}(P_{XY})$,
		\item [iii.] $h_{\epsilon}(P_{XY})=H(Y|X)+\epsilon$.
	\end{itemize}
\end{theorem}
\begin{proof}
 The statements i $\Rightarrow$ ii and iii $\Rightarrow$ i can be shown by using Lemma~\ref{goh} and Lemma~\ref{kir}, respectively. For proving ii $\Rightarrow$ iii, let $\bar{U}$ be an optimizer of $g_{\epsilon}(P_{XY})$, by using Lemma~\ref{kir}, \eqref{key} and Markov chain $X-Y-U$ we show that $I(\bar{U};Y)=I(X;\bar{U})+H(Y|X)$. Furthermore, we show that $I(X;\bar{U})=\epsilon$, thus, $h_{\epsilon}(P_{XY})=H(Y|X)+\epsilon$. A detailed proof is provided in Appendix B.   
\end{proof}
By using Theorem~\ref{ziba} and Proposition~\ref{sef}, Corollary~\ref{seff} can be strengthened as follows.
\begin{corollary}\label{sefff}
	If the common information and mutual information between $X$ and $Y$ are equal then we have
	\begin{align*}
	g_{\epsilon}(P_{XY})=h_{\epsilon}(P_{XY})=H(Y|X)+\epsilon.
	\end{align*}
\end{corollary}
Lemma~\ref{kir} and Theorem~\ref{ziba} generalize \cite[Th.~7]{kostala} for non-zero leakage. Now let $\epsilon=0$ in Lemma~\ref{kir} and Theorem~\ref{ziba}. As argued in \cite{kostala} when $X$ is a deterministic function of $Y$, the necessary and sufficient conditions for having equality in Lemma~\ref{kir} are fulfilled. Furthermore, this result holds when $X$ and $Y$ are independent or $Y$ is a deterministic function of $X$. However, in this work we have shown that for any $0\leq \epsilon < I(X;Y)$, this statement can be generalized and we can substitute the condition that $X$ is a deterministic function of $Y$ by the condition that the common information and mutual information between $X$ and $Y$ are equal.  
\subsection* {Special case: $\epsilon=0$ (Independent $X$ and $U$)}
In this section we derive new lower and upper bounds for $h_0(P_{XY})$ and compare them with the previous bounds found in \cite{kostala}. We first state the definition of \emph{excess functional information} defined in
\cite{kosnane} as
\begin{align*}
\psi(X\rightarrow Y)=\inf_{\begin{array}{c} 
	\substack{P_{U|Y,X}: I(U;X)=0,\ H(Y|X,U)=0}
	\end{array}}I(X;U|Y),
\end{align*}
and the lower bound on $\psi(X\rightarrow Y)$ derived in \cite[Prop.~1]{kosnane} is given in the next lemma. Since this lemma is useful for deriving the upper bound on $h_{\epsilon}(P_{XY})$ we state it here.
\begin{lemma}\cite[Prop.~1]{kosnane}\label{haroomi}
	For discrete $Y$ we have 
	\begin{align}
	\psi(X\rightarrow Y)\geq -\sum_{y\in\mathcal{Y}}\!\int_{0}^{1}\!\!\! \mathbb{P}_X\{P_{Y|X}(y|X)\geq t\}\log (\mathbb{P}_X\{P_{Y|X}(y|X)\geq t\})dt-I(X;Y),\label{koonsher}
	\end{align}
	where for $|\mathcal{Y}|=2$ the equality holds and it is attained by the Poisson functional representation in \cite{kosnane}.
\end{lemma}
\begin{remark}
	The lower bound in \eqref{koonsher} can be negative. For instance, let $Y$ be a deterministic function of $X$, i.e., $H(Y|X)=0$. In this case we have $-\sum_{y\in\mathcal{Y}}\!\int_{0}^{1} \mathbb{P}_X\{P_{Y|X}(y|X)\geq t\}\log (\mathbb{P}_X\{P_{Y|X}(y|X)\geq t\})dt-I(X;Y)=-I(X;Y)=-H(Y).$
\end{remark}
In the next theorem lower and upper bounds on $h_0(P_{XY})$ are provided.
\begin{theorem}\label{khata}
	For any pair of RVs $(X,Y)$ distributed according to $P_{XY}$ supported on alphabets $\mathcal{X}$ and $\mathcal{Y}$ we have
	\begin{align*}
	\max\{L^0_1,L^0_2\} \leq h_{0}(P_{XY})\leq \min\{U^0_1,U^0_2\},
	\end{align*}
	where $L^0_1$ and $L^0_2$ are defined in Corollary~\ref{kooni} and 
	\begin{align*}
	&U^0_1 = H(Y|X),\\
	&U^0_2 =  H(Y|X) +\sum_{y\in\mathcal{Y}}\int_{0}^{1} \mathbb{P}_X\{P_{Y|X}(y|X)\geq t\}\log (\mathbb{P}_X\{P_{Y|X}(y|X)\geq t\})dt+I(X;Y).
	\end{align*}	
	Furthermore, if $|\mathcal{Y}|=2$, then we have 
	\begin{align*}
	h_{0}(P_{XY}) = U^0_2.
	\end{align*}
\end{theorem}
\begin{proof}
	$L^0_1$ and $L^0_2$ can be obtained by letting $\epsilon=0$ in Theorem~\ref{th.1}. $U^0_1$ which has been derived in \cite[Th.~7]{kostala} can be obtained by \eqref{key}. $U^0_1$ can be derived as follows. Since $X$ and $U$ are independent, \eqref{key} can be rewritten as 
	\begin{align*}
	I(Y;U)=H(Y|X)-H(Y|U,X)-I(X;U|Y),
	\end{align*} 
	thus, using Lemma~\ref{haroomi}
	\begin{align*}
	h_0(P_{XY})&\leq H(Y|X)-\inf_{H(Y|U,X)=0,\ I(X;U)=0}I(X;U|Y)\\ &= H(Y|X)-\psi(X\rightarrow Y)\\&\leq H(Y|X)+\sum_{y\in\mathcal{Y}}\int_{0}^{1} \mathbb{P}_X\{P_{Y|X}(y|X)\geq t\}\log (\mathbb{P}_X\{P_{Y|X}(y|X)\geq t\})dt+I(X;Y). 	 \end{align*}
	For $|\mathcal{Y}|=2$ using Lemma~\ref{haroomi} we have $\psi(X\rightarrow Y)=-\sum_{y\in\mathcal{Y}}\int_{0}^{1} \mathbb{P}_X\{P_{Y|X}(y|X)\geq t\}\log (\mathbb{P}_X\{P_{Y|X}(y|X)\geq t\})dt-I(X;Y)$ and let $\bar{U}$ be the RV that attains this bound. Thus,
	\begin{align*}
	I(\bar{U};Y)&=H(Y|X)+\sum_{y\in\mathcal{Y}}\int_{0}^{1} \mathbb{P}_X\{P_{Y|X}(y|X)\geq t\}\log (\mathbb{P}_X\{P_{Y|X}(y|X)\geq t\})dt+I(X;Y).
	\end{align*}
	Therefore, $\bar{U}$ attains $U_2^0$ and $h_0(P_{XY})=U_0^2$.
\end{proof}
As mentioned before the upper bound $U^0_1$ has been derived in \cite[Th.~7]{kostala}. The upper bound $U^0_2$ is a new upper bound. Thus, the lower and upper bounds on $h_{0}(P_{XY})$ stated in Theorem~\ref{khata} generalize the bounds in \cite{kostala}. Furthermore, in case of binary $Y$ the exact expression for $h_{0}(P_{XY})$ has been derived.
\begin{lemma}\label{ankhar}
	If $X$ is a deterministic function of $Y$, i.e., $H(X|Y)=0$, we have
	\begin{align*}
	&\sum_{y\in\mathcal{Y}}\int_{0}^{1} \mathbb{P}_X\{P_{Y|X}(y|X)\geq t\}\log (\mathbb{P}_X\{P_{Y|X}(y|X)\geq t\})dt+I(X;Y)=0.
	\end{align*}
\end{lemma}
\begin{proof}
	The proof is provided in Appendix B.
\end{proof}
\begin{remark}
	According to Lemma~\ref{ankhar}, if $X$ is a deterministic function of $Y$, then we have $U_2^0=U_1^0$.  
\end{remark}
In the next example we compare the bounds $U^0_1$ and $U^0_2$ for a $BSC(\theta)$.
\begin{figure}[]
	\centering
	\includegraphics[scale = .2]{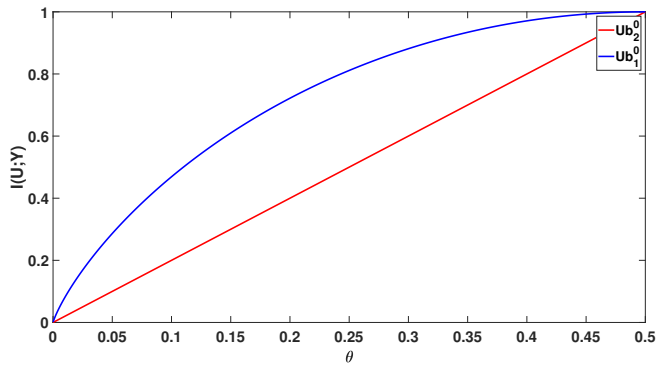}
	\caption{Comparing the upper bounds $U_1^0$ and $U_2^0$ for $BSC(\theta)$. The blue curve illustrates the upper bound found in \cite{kostala} and the red line shows the upper bound found in this work. In fact by using Theorem~\ref{khata} the red curve corresponding to $U_2^0$ can be achieved and it presents the solution for $h_{0}(P_{XY})$.}
	\label{fig:kir}
\end{figure}
\begin{example}(Binary Symmetric Channel)
	Let the binary RVs $X\in\{0,1\}$ and $Y\in\{0,1\}$ have the following joint distribution
	\begin{align*}
	P_{XY}(x,y)=\begin{cases}
	\frac{1-\theta}{2}, \ &x=y\\
	\frac{\theta}{2}, \ &x\neq y
	\end{cases},
	\end{align*}
	where $\theta<\frac{1}{2}$. We obtain
	\begin{align*}
	\sum_{y\in\mathcal{Y}}\int_{0}^{1} \mathbb{P}_X\{P_{Y|X}(y|X)\geq t\}\log (\mathbb{P}_X\{P_{Y|X}(y|X)\geq t\})dt&=\int_{\theta}^{1-\theta} \mathbb{P}_X\{P_{Y|X}(0|X)\geq t\}\log (\mathbb{P}_X\{P_{Y|X}(0|X)\geq t\})dt\\&+\int_{\theta}^{1-\theta} \mathbb{P}_X\{(P_{Y|X}(1|X)\geq t\}\log (\mathbb{P}_X\{P_{Y|X}(1|X)\geq t\})dt\\&=(1-2\theta)\left(P_X(0)\log (P_X(0))+P_X(1)\log (P_X(1))\right)\\&=-(1-2\theta)H(X)=-(1-2\theta).
	\end{align*}
	Thus,
	\begin{align*}
	U_2^0&=H(Y|X) +\sum_{y\in\mathcal{Y}}\int_{0}^{1} \mathbb{P}_X\{P_{Y|X}(y|X)\geq t\}\log (\mathbb{P}_X\{P_{Y|X}(y|X)\geq t\})dt+I(X;Y)\\&=h(\theta)-(1-2\theta)+(1-h(\theta))=2\theta,\\
	U_1^0&=H(Y|X)=h(\theta),
	\end{align*}
	where $h(\cdot)$ corresponds to the binary entropy function. As shown in Fig.~ \ref{fig:kir}, we have
	\begin{align*}
	h_{0}(P_{XY})\leq U^0_2\leq U^0_1.
	\end{align*}
	However by using Theorem~\ref{khata}, since $|\mathcal{Y}|=2$ the upper bound $U_2^0$ is achieved and we have
	\begin{align*}
	h_{0}(P_{XY})= U^0_2.
	\end{align*}
\end{example}
\begin{example}(Erasure Channel)
	Let the RVs $X\in\{0,1\}$ and $Y\in\{0,e,1\}$ have the following joint distribution
	\begin{align*}
	P_{XY}(x,y)=\begin{cases}
	\frac{1-\theta}{2}, \ &x=y\\
	\frac{\theta}{2}, \ &y=e\\
	0, \ & \text{else}
	\end{cases},
	\end{align*}
	where $\theta<\frac{1}{2}$. We have
	\begin{align*}
	&\sum_{y\in\mathcal{Y}}\int_{0}^{1} \mathbb{P}_X\{P_{Y|X}(y|X)\geq t\}\log (\mathbb{P}_X\{P_{Y|X}(y|X)\geq t\})dt\\&=\int_{0}^{1-\theta} \mathbb{P}_X\{P_{Y|X}(0|X)\geq t\}\log (\mathbb{P}_X\{P_{Y|X}(0|X)\geq t\})dt\\&+\int_{0}^{1-\theta} \mathbb{P}_X\{P_{Y|X}(1|X)\geq t\}\log (\mathbb{P}_X\{P_{Y|X}(1|X)\geq t\})dt\\&=-(1-\theta)H(X)=-(1-\theta).
	\end{align*}
	Thus, 
	\begin{align*}
	U_2^0&=H(Y|X) +\sum_{y\in\mathcal{Y}}\int_{0}^{1} \mathbb{P}_X\{P_{Y|X}(y|X)\geq t\}\log (\mathbb{P}_X\{P_{Y|X}(y|X)\geq t\})dt+I(X;Y)\\&=h(\theta)-(1-\theta)+h(\theta)+1-\theta-h(\theta) \\&=h(\theta),\\
	U_1^0&=H(Y|X)=h(\theta). 
	\end{align*}
	Hence, in this case, $U_1^0=U_2^0=h(\theta)$. Furthermore, in \cite[Example~8]{kostala}, it has been shown that for this pair of $(X,Y)$ we have $g_0(P_{XY})=h_0(P_{XY})=h(\theta)$.
\end{example}
In \cite[Prop.~2]{kosnane} it has been shown that for every $\alpha\geq 0$, there exist a pair $(X,Y)$ such that $I(X;Y)\geq \alpha$ and 
\begin{align}
\psi(X\rightarrow Y)\geq \log(I(X;Y)+1)-1.\label{kir1}
\end{align}
\begin{lemma}\label{choon}
	Let $(X,Y)$ be as in \cite[Prop.~2]{kosnane}, i.e. $(X,Y)$ satisfies \eqref{kir1}. Then for such pair we have
	\begin{align*}
	H(Y|X)-\log(I(X;Y)+1)-4\leq h_0(P_{XY})\leq H(Y|X)-\log(I(X;Y)+1)+1. 
	\end{align*}
	\begin{proof}
		The lower bound follows from Corollary~1. For the upper bound, we use \eqref{key} and \eqref{kir1} so that
		\begin{align*}
		I(U;Y)&\leq H(Y|X)-\psi(X\rightarrow Y)\\ &\leq H(Y|X)-\log(I(X;Y)+1)+1.
		\end{align*}
	\end{proof}
\end{lemma}
\begin{remark}
	From Lemma~\ref{choon} and Corollary~1 we can conclude that the lower bound $L_2^0=H(Y|X)-(\log(I(X;Y)+1)+4)$ is tight within $5$ bits.
\end{remark}
\subsection{Privacy-utility trade-off with non-zero leakage and per-letter privacy constraints}
In this section, we first introduce similar lemmas as Lemma~\ref{lemma3} and Lemma~\ref{lemma4}, where we have replaced the mutual information constraint, i.e., $I(U;X)=\epsilon$, with the strong privacy criterion 1 defined in \eqref{main22} and \eqref{main11}. In the remaining part of this work $d(\cdot,\cdot)$ corresponds to the total variation distance, i.e., $d(P,Q)=\sum_x |P(x)-Q(x)|$. 
\begin{lemma}\label{lemma11}
	For any $0\leq\epsilon< \sqrt{2I(X;Y)}$ and any pair of RVs $(X,Y)$ distributed according to $P_{XY}$ supported on alphabets $\mathcal{X}$ and $\mathcal{Y}$ where $|\mathcal{X}|$ is finite and $|\mathcal{Y}|$ is finite or countably infinite, there exists a RV $U$ supported on $\mathcal{U}$ such that $X$ and $U$ satisfy the strong privacy criterion 1, i.e., we have
	\begin{align}\label{c11}
	d(P_{X,U}(\cdot,u),P_XP_{U}(u))\leq\epsilon,\ \forall u,
	\end{align}
	$Y$ is a deterministic function of $(U,X)$, i.e., we have
	\begin{align}
	H(Y|U,X)=0,\label{c21}
	\end{align}
	and 
	\begin{align}
	|\mathcal{U}|\leq \left[|\mathcal{X}|(|\mathcal{Y}|-1)+1\right]\left[|\mathcal{X}|+1\right].\label{c31}
	\end{align}
\end{lemma}
\begin{proof}
	Let $U$ be found by the EFRL, where we let the leakage be $\frac{\epsilon^2}{2}$. Thus, we have
	\begin{align*}
	\frac{\epsilon^2}{2}&=I(U;X)\\&=\sum_u P_U(u)D(P_{X|U}(\cdot|u),P_X)\\ &\stackrel{(a)}{\geq} \sum_u \frac{P_U(u)}{2}\left( d(P_{X|U}(\cdot|u),P_X)\right)^2\\&\stackrel{(b)}{\geq} \sum_u \frac{P_U(u)^2}{2}\left( d(P_{X|U}(\cdot|u),P_X)\right)^2\\ &\geq \frac{P_U(u)^2}{2}\left( d(P_{X|U}(\cdot|u),P_X)\right)^2\\ &= \frac{\left( d(P_{X,U}(\cdot,u),P_XP_U(u))\right)^2}{2},
	\end{align*}
	where $D(\cdot,\cdot)$ corresponds to KL-divergence. Furthermore, (a) follows by the Pinsker's inequality \cite{verdu} and (b) follows since $0\leq P_U(u)\leq 1$. Using the last line we obtain
	\begin{align*}
	d(P_{X,U}(\cdot,u),P_XP_U(u))\leq \epsilon,\ \forall u.
	\end{align*}
	The other constraints can be obtained by using Lemma~\ref{lemma3}.	
\end{proof}
\begin{remark}
	RV $U$, which is specified by the FRL (Lemma~\ref{lemma1}), satisfies all constraints in Lemma~\ref{lemma11}. However, as we show later, it achieves less utility compared to the RV $U$ which is used in the proof of Lemma~\ref{lemma11}. Furthermore, we can add constraints such as $0<I(U;X)$ and $0<\epsilon< \sqrt{2I(X;Y)}$ to Lemma~\ref{lemma11} while the RV $U$ found by the FRL does not satisfy them. 
\end{remark}
\begin{lemma}\label{lemma22} 
	For any $0\leq\epsilon< \sqrt{2I(X;Y)}$ and pair of RVs $(X,Y)$ distributed according to $P_{XY}$ supported on alphabets $\mathcal{X}$ and $\mathcal{Y}$ where $|\mathcal{X}|$ is finite and $|\mathcal{Y}|$ is finite or countably infinite with $I(X,Y)< \infty$, there exists a RV $U$ defined on $\mathcal{U}$ such that $X$ and $U$ satisfy the strong privacy criterion 1, i.e., we have
	\begin{align*}
	d(P_{X,U}(\cdot,u),P_XP_{U}(u))\leq\epsilon,\ \forall u,
	\end{align*}
	$Y$ is a deterministic function of $(U,X)$, i.e., we have 
	\begin{align*}
	H(Y|U,X)=0,
	\end{align*}
	$I(X;U|Y)$ can be  upper bounded as follows 
	\begin{align}
	I(X;U|Y)\!\leq \alpha H(X|Y)\!+\!(1-\alpha)\!\left[ \log(I(X;Y)+1)+4\right],\label{bala}
	\end{align}
	and 
	$
	|\mathcal{U}|\leq \left[|\mathcal{X}|(|\mathcal{Y}|-1)+2\right]\left[|\mathcal{X}|+1\right],
	$
	where $\alpha =\frac{\epsilon^2}{2H(X)}$.
\end{lemma}
\begin{proof}
	Let $U$ be found by the ESFRL, where we let the leakage be $\frac{\epsilon^2}{2}$. The first constraint in this statement can be obtained by using the same proof as Lemma \ref{lemma11}. Furthermore, \eqref{bala} can be derived using Lemma~\ref{lemma4}. 
\end{proof}
\begin{remark}
	RV $U$ produced by the SFRL (Lemma~\ref{lemma2}) does not satisfy \eqref{bala} in general. However, in case of satisfying \eqref{bala}, by using similar arguments as Corollary~\ref{sos}, it achieves less or equal utility compared to the RV $U$ which is used in the proof of Lemma~\ref{lemma22}. Similar to Corollary~\ref{sos}, we later show that the RV found by proof of Lemma~\ref{lemma22} strictly improves the utility for non-zero leakage. 
\end{remark}
In the next proposition we find a lower bound on $h_{\epsilon}^{w\ell}(P_{XY})$ using Lemma~1 and Lemma~2. 
\begin{proposition}\label{prop111}
	For any $0\leq \epsilon< \sqrt{2I(X;Y)}$ and pair of RVs $(X,Y)$ distributed according to $P_{XY}$ supported on alphabets $\mathcal{X}$ and $\mathcal{Y}$ we have
	\begin{align}\label{prop12}
	h_{\epsilon}^{w\ell}(P_{XY})\geq \max\{L_{h^{w\ell}}^{1}(\epsilon),L_{h^{w\ell}}^{2}(\epsilon)\},
	\end{align}
	where
	\begin{align*}
	L_{h^{w\ell}}^{1}(\epsilon) &= H(Y|X)-H(X|Y)+\frac{\epsilon^2}{2},\\
	L_{h^{w\ell}}^{2}(\epsilon) &= H(Y|X)-\alpha H(X|Y)+\frac{\epsilon^2}{2} -(1-\alpha)\left( \log(I(X;Y)+1)+4 \right),\\
	\end{align*}
	with $\alpha=\frac{\epsilon^2}{2H(X)}$.
\end{proposition}
\begin{proof}
	For deriving $L_{h^{w\ell}}^{1}(\epsilon)$ let $U$ be produced as in the proof of Lemma \ref{lemma11}. Thus, $I(X;U)=\frac{\epsilon^2}{2}$ and $U$ satisfies \eqref{c11} and \eqref{c21}. We have
	\begin{align*}
	h_{\epsilon}^{w\ell}(P_{XY})&\geq
	I(U;Y)\\&=I(X;U)\!+\!H(Y|X)\!-\!I(X;U|Y)\!-\!H(Y|X,U)\\&=\frac{\epsilon^2}{2}+H(Y|X)-H(X|Y)+H(X|Y,U)\\ &\geq \frac{\epsilon^2}{2}+H(Y|X)-H(X|Y).
	\end{align*}
	Next for deriving $L_{h^{w\ell}}^{2}(\epsilon)$ let $U$ be produced by Lemma \ref{lemma22}. Hence, $I(X;U)=\frac{\epsilon^2}{2}$ and $U$ satisfies \eqref{c11}, \eqref{c21}, and \eqref{bala}. We obtain
	\begin{align*}
	h_{\epsilon}^{w\ell}(P_{XY})&\geq I(U;Y) = \frac{\epsilon^2}{2}+H(Y|X)-I(X;U|Y)\\ &\geq \frac{\epsilon^2}{2}+H(Y|X)-\alpha H(X|Y)-(1-\alpha)\left( \log(I(X;Y)+1)+4 \right).
	\end{align*}
\end{proof}
In the next section, we provide a lower bound on $g_{\epsilon}^{w\ell}(P_{XY})$ by following the same approach as in \cite{Khodam22}. For more details about the proofs and steps of approximation see \cite[Section III]{Khodam22}.
\subsection*{Lower bound on $g_{\epsilon}^{w\ell}(P_{XY})$}
In \cite{Khodam22}, we show that $g_{\epsilon}^{\ell}(P_{XY})$ can be approximated by a linear program using information geometry concepts. Using this result we can derive a lower bound for $g_{\epsilon}^{\ell}(P_{XY})$. In this part, we follow a similar approach to approximate $g_{\epsilon}^{w\ell}(P_{XY})$, which results in a lower bound. Similar to \cite{Khodam22}, for sufficiently small $\epsilon$, by using the leakage constraint in $g_{\epsilon}^{\ell}(P_{XY})$, i.e., the strong privacy criterion 1, we can rewrite the distribution $P_{X,U}(\cdot,u)$ as a perturbation of $P_XP_U(u)$. Thus, for any $u$ we can write $P_{X,U}(\cdot,u)=P_XP_U(u)+\epsilon J_u$, where $J_u\in \mathbb{R}^{|\mathcal{X}|}$ is a perturbation vector and satisfies the following properties:
\begin{align}
\bm{1}^T\cdot J_u&=0,\ \forall u, \label{koon1}\\
\sum_u J_u&=\bm{0}\in \mathbb{R}^{|\mathcal{X}|},\label{koon2}\\
\bm{1}^T\cdot |J_u|&\leq 1,\ \forall u, \label{koon3}
\end{align} 
where $|\cdot|$ corresponds to the absolute value of the vector. 
We define matrix $M\in \mathbb{R}^{|\mathcal{X}|\times|\mathcal{Y}|}$, which is used in the remaining part, as follows: Let $V$ be the matrix of right eigenvectors of $P_{X|Y}$, i.e., $P_{X|Y}=U\Sigma V^T$ and $V=[v_1,\ v_2,\ ... ,\ v_{|\mathcal{Y}|}]$, then $M$ is defined as
\begin{align*}
M \triangleq \left[v_1,\ v_2,\ ... ,\ v_{|\mathcal{X}|}\right]^T.  
\end{align*}  
Similar to \cite[Proposition~2]{Khodam22}, we have the following result.
\begin{proposition}\label{prop222}
	In \eqref{main2}, it suffices to consider $U$ such that $|\mathcal{U}|\leq|\mathcal{Y}|$. Since the supremum in \eqref{main2} is achieved, we can replace the supremum by the maximum.
\end{proposition}
\begin{proof}
	The proof follows the similar lines as the proof of \cite[Proposition~2]{Khodam22}. The only difference is that the new convex and compact set is as follows
	\begin{align*}
	\Psi\!=\!\left\{\!y\in\mathbb{R}_{+}^{|\mathcal{Y}|}|My\!=\!MP_Y\!+\!\frac{\epsilon}{P_U(u)} M\!\!\begin{bmatrix}
	P_{X|Y_1}^{-1}J_u\\0
	\end{bmatrix}\!\!,J_u\in\mathcal{J} \!\right\}\!,
	\end{align*}
	where $\mathcal{J}=\{J\in\mathbb{R}^{|\mathcal{X}|}_{+}|\left\lVert J\right\rVert_1\leq 1,\ \bm{1}^{T}\cdot J=0\}$ and $\mathbb{R}_{+}$ corresponds to non-negative real numbers. Only non-zero weights $P_U(u)$ are considered since in the other case the corresponding $P_{Y|U}(\cdot|u)$ does not appear in $H(Y|U)$. 
\end{proof}
\begin{lemma}\label{madar1}
	If the Markov chain $X-Y-U$ holds, for sufficiently small $\epsilon$ and every $u\in\mathcal{U}$, the vector $P_{Y|U}(\cdot|u)$ lies in the following convex polytope
	\begin{align*}
	\mathbb{S}_{u} = \left\{y\in\mathbb{R}_{+}^{|\mathcal{Y}|}|My=MP_Y+\frac{\epsilon}{P_U(u)} M\begin{bmatrix}
	P_{X|Y_1}^{-1}J_u\\0
	\end{bmatrix}\right\},
	\end{align*}
	where $J_u$ satisfies \eqref{koon1}, \eqref{koon2}, and \eqref{koon3}. Furthermore, $P_U(u)>0$, otherwise $P_{Y|U}(\cdot|u)$ does not appear in $I(Y;U)$.
\end{lemma}
\begin{proof}
	Using the Markov chain $X-Y-U$, we have
	\begin{align*}
	P_{X|U=u}-P_X=P_{X|Y}[P_{Y|U=u}-P_Y]=\epsilon \frac{J_u}{P_U(u)}.
	\end{align*}
	Thus, by following the similar lines as \cite[Lemma~2]{Khodam22} and using the properties of Null($M$) as \cite[Lemma~1]{Khodam22}, we have
	\begin{align*}
	MP_{Y|U}(\cdot|u)=MP_Y+\frac{\epsilon}{P_U(u)} M\begin{bmatrix}
	P_{X|Y_1}^{-1}J_u\\0
	\end{bmatrix}.
	\end{align*} 
\end{proof}
By using the same arguments as \cite[Lemma~3]{Khodam22}, it can be shown that any vector inside $\mathbb{S}_{u}$ is a standard probability vector. Thus, by using \cite[Lemma~3]{Khodam22} and Lemma~2 we have following result.
\begin{theorem}
	We have the following equivalency 
	\begin{align}\label{equii}
	\min_{\begin{array}{c} 
		\substack{P_{U|Y}:X-Y-U\\ d(P_{X,U}(\cdot,u),P_XP_U(u))\leq\epsilon,\ \forall u\in\mathcal{U}}
		\end{array}}\! \! \! \!\!\!\!\!\!\!\!\!\!\!\!\!\!\!\!H(Y|U) =\!\!\!\!\!\!\!\!\! \min_{\begin{array}{c} 
		\substack{P_U,\ P_{Y|U=u}\in\mathbb{S}_u,\ \forall u\in\mathcal{U},\\ \sum_u P_U(u)P_{Y|U=u}=P_Y,\\ J_u \text{satisfies}\ \eqref{koon1},\ \eqref{koon2},\ \text{and}\ \eqref{koon3}}
		\end{array}} \!\!\!\!\!\!\!\!\!\!\!\!\!\!\!\!\!\!\!H(Y|U).
	\end{align}
\end{theorem}  
Furthermore, similar to \cite[Proposition~3]{Khodam22}, it can be shown that the minimum of $H(Y|U)$ occurs at the extreme points of the sets $\mathbb{S}_{u}$, i.e., for each $u\in \mathcal{U}$, $P_{Y|U}^*(\cdot|u)$ that minmizes $H(Y|U)$ must belong to the extreme points of $\mathbb{S}_{u}$. To find the extreme points of $\mathbb{S}_{u}$ let $\Omega$ be the set of indices which correspond to $|\mathcal{X}|$ linearly independent columns of $M$, i.e., $|\Omega|=|\mathcal{X}|$ and $\Omega\subset \{1,..,|\mathcal{Y}|\}$. Let $M_{\Omega}\in\mathbb{R}^{|\mathcal{X}|\times|\mathcal{X}|}$ be the submatrix of $M$ with columns indexed by the set $\Omega$. Assume that $\Omega = \{\omega_1,..,\omega_{|\mathcal{X}|}\}$, where $\omega_i\in\{1,..,|\mathcal{Y}|\}$ and all elements are arranged in an increasing order. The $\omega_i$-th element of the extreme point $V_{\Omega}^*$ can be found as $i$-th element of $M_{\Omega}^{-1}(MP_Y+\frac{\epsilon}{P_U(u)} M\begin{bmatrix}
P_{X|Y_1}^{-1}J_u\\0\end{bmatrix})$, i.e., for $1\leq i \leq |\mathcal{X}|$ we have
\begin{align}\label{defin1}
V_{\Omega}^*(\omega_i)= \left(M_{\Omega}^{-1}MP_Y+\frac{\epsilon}{P_U(u)} M_{\Omega}^{-1}M\begin{bmatrix}
P_{X|Y_1}^{-1}J_u\\0\end{bmatrix}\right)(i).
\end{align}
Other elements of $V_{\Omega}^*$ are set to be zero. Now we approximate the entropy of $V_{\Omega}^*$.
\begin{proposition}\label{koonkos}
	Let $V_{\Omega_u}^*$ be an extreme point of the set $\mathbb{S}_u$, then we have
	\begin{align*}
	H(P_{Y|U=u}) &=\sum_{y=1}^{|\mathcal{Y}|}-P_{Y|U=u}(y)\log(P_{Y|U=u}(y))\\&=-(b_u+\frac{\epsilon}{P_{U}(u)} a_uJ_u)+o(\epsilon),
	\end{align*}
	with $b_u = l_u \left(M_{\Omega_u}^{-1}MP_Y\right),\ 
	a_u = l_u\left(M_{\Omega_u}^{-1}M(1\!\!:\!\!|\mathcal{X}|)P_{X|Y_1}^{-1}\right)\in\mathbb{R}^{1\times|\mathcal{X}|},\
	l_u = \left[\log\left(M_{\Omega_u}^{-1}MP_{Y}(i)\right)\right]_{i=1:|\mathcal{X}|}\in\mathbb{R}^{1\times|\mathcal{X}|},
	$ and $M_{\Omega_u}^{-1}MP_{Y}(i)$ stands for $i$-th ($1\leq i\leq |\mathcal{X}|$) element of the vector $M_{\Omega_u}^{-1}MP_{Y}$. Furthermore, $M(1\!\!:\!\!|\mathcal{X}|)$ stands for submatrix of $M$ with first $|\mathcal{X}|$ columns.
\end{proposition}
\begin{proof}
	The proof follows similar lines as \cite[Lemma~4]{Khodam22} and is based on first order Taylor expansion of $\log(1+x)$.
\end{proof}
By using Proposition \ref{koonkos} we can approximate \eqref{main2} as follows.
\begin{proposition}\label{baghal}
	For sufficiently small $\epsilon$, the minimization problem in \eqref{equii} can be approximated as follows
	\begin{align}\label{kospa}
	&\min_{P_U(.),\{J_u, u\in\mathcal{U}\}} -\left(\sum_{u=1}^{|\mathcal{Y}|} P_U(u)b_u+\epsilon a_uJ_u\right)\\\nonumber
	&\text{subject to:}\\\nonumber
	&\sum_{u=1}^{|\mathcal{Y}|} P_U(u)V_{\Omega_u}^*=P_Y,\ \sum_{u=1}^{|\mathcal{Y}|} J_u=0,\ P_U\in \mathbb{R}_{+}^{|\cal Y|},\\\nonumber
	&\bm{1}^T |J_u|\leq 1,\  \bm{1}^T\cdot J_u=0,\ \forall u\in\mathcal{U},
	\end{align} 
	where $a_u$ and $b_u$ are defined in Proposition~\ref{koonkos}.
\end{proposition}
\begin{proof}
	The proof follows directly from Proposition~\ref{koonkos} and the fact that the minimum of $H(Y|U)$ occurs at the extreme points of the sets $\mathbb{S}_u$. Thus, for $P_{Y|U=u}=V_{\Omega_u}^*,\ u\in\{1,..,|\mathcal{Y}|\}$, where $V_{\Omega_u}^*$ is defined in \eqref{defin1}, $H(Y|U)$ can be approximated as follows
	\begin{align*}
	H(Y|U) \!= \sum_u\! P_U(u)H(P_{Y|U=u})\!\cong \sum_{u=1}^{|\mathcal{Y}|}\! P_U(u)b_u\!+\!\epsilon a_uJ_u.
	\end{align*} 
\end{proof}
By using the vector
$\eta_u=P_U(u)\left(M_{\Omega_u}^{-1}MP_Y\right)+\epsilon \left(M_{\Omega_u}^{-1}M(1:|\mathcal{X}|)P_{X|Y_1}^{-1}\right)(J_u)$ for all $u\in \mathcal{U}$, where $\eta_u\in\mathbb{R}^{|\mathcal{X}|}$, we can write \eqref{kospa} as a linear program. The vector $\eta_u$ corresponds to multiple of non-zero elements of the extreme point $V_{\Omega_u}^*$, furthermore, $P_U(u)$ and $J_u$ can be uniquely found as
\begin{align*}
P_U(u)&=\bm{1}^T\cdot \eta_u,\\
J_u&=\frac{P_{X|Y_1}M(1:|\mathcal{X}|)^{-1}M_{\Omega_u}[\eta_u\!-\!(\bm{1}^T \eta_u)M_{\Omega_u}^{-1}MP_Y]}{\epsilon}.
\end{align*}
By solving the linear program we obtain $P_U$ and $J_u$ for all $u$, thus, $P_{Y|U}(\cdot|u)$ can be computed using \eqref{defin1}.
\begin{lemma}\label{lg1}
	Let $P_{U|Y}^*$ be found by the linear program, which solves \eqref{kospa}, and let $I(U^*;Y)$ be evaluated by this kernel. Then we have
	\begin{align*}
	g_{\epsilon}^{w\ell}(P_{XY})\geq I(U^*;Y) = L_{g^{w\ell}}^1(\epsilon).
	\end{align*}
\end{lemma}
\begin{proof}
	The proof directly follows since the kernel $P_{U|Y}^*$ that achieves the approximate solution satisfies the constraints in \eqref{main2}.
\end{proof}
In the next result we present lower and upper bounds of $g_{\epsilon}^{w\ell}(P_{XY})$ and $h_{\epsilon}^{w\ell}(P_{XY})$.
\begin{theorem}\label{choon1}
	For sufficiently small $\epsilon\geq 0$ and any pair of RVs $(X,Y)$ distributed according to $P_{XY}$ supported on alphabets $\mathcal{X}$ and $\mathcal{Y}$ we have
	\begin{align*}
	L_{g^{w\ell}}^1(\epsilon)\leq g_{\epsilon}^{w\ell}(P_{XY}),
	\end{align*}	
	and for any $\epsilon\geq 0$ we obtain
	\begin{align*}
	g_{\epsilon}^{w\ell}(P_{XY})&\leq \frac{\epsilon|\mathcal{Y}||\mathcal{X}|}{\min P_X}+H(Y|X)=U_{g^{w\ell}}(\epsilon),\\
	g_{\epsilon}^{w\ell}(P_{XY})&\leq h_{\epsilon}^{w\ell}(P_{XY}).
	\end{align*}
	Furthermore, for any $0\leq \epsilon\leq \sqrt{2I(X;Y)}$ we have
	\begin{align*}
	\max\{L_{h^{w\ell}}^{1}(\epsilon),L_{h^{w\ell}}^{2}(\epsilon)\}\leq h_{\epsilon}^{w\ell}(P_{XY}),
	\end{align*}
	where $L_{h^{w\ell}}^{1}(\epsilon)$ and $L_{h^{w\ell}}^{2}(\epsilon)$ are defined in Proposition~\ref{prop111}.
\end{theorem}
\begin{proof}
	The proof is provided in Appendix C.
\end{proof}
In the next section we provide bounds for $g_{\epsilon}^{\ell}(P_{XY})$ and $h_{\epsilon}^{\ell}(P_{XY})$.
\subsection*{Lower and Upper bounds on $g_{\epsilon}^{\ell}(P_{XY})$ and $h_{\epsilon}^{\ell}(P_{XY})$}
As we mentioned earlier in \cite{Khodam22}, we have provided an approximate solution for $g_{\epsilon}^{\ell}(P_{XY})$ using a local approximation of $H(Y|U)$ for sufficiently small $\epsilon$. Furthermore, in \cite[Proposition~8]{Khodam22} we specified permissible leakages. By using \cite[Proposition~8]{Khodam22}, we can write
\begin{align}
g_{\epsilon}^{\ell}(P_{XY})&=\sup_{\begin{array}{c} 
	\substack{P_{U|Y}:X-Y-U\\ \ d(P_{X|U}(\cdot|u),P_X)\leq\epsilon,\ \forall u}
	\end{array}}I(Y;U)\nonumber\\&= \max_{\begin{array}{c} 
	\substack{P_{U|Y}:X-Y-U\\ \ d(P_{X|U}(\cdot|u),P_X)\leq\epsilon,\ \forall u\\ |\mathcal{U}|\leq |\mathcal{Y}|}
	\end{array}}I(Y;U).\label{antar11}
\end{align}
In the next lemma we find a lower bound for $g_{\epsilon}^{\ell}(P_{XY})$, where we use the approximate problem for \eqref{main222}.
\begin{lemma}\label{lg2}
	Let the kernel $P_{U^*|Y}$ achieve the optimum solution in \cite[Theorem~2]{Khodam22}. Thus, $I(U^*;Y)$ evaluated by this kernel is a lower bound for $g_{\epsilon}^{\ell}(P_{XY})$. In other words, we have
	\begin{align*}
	g_{\epsilon}^{\ell}(P_{XY})\geq I(U^*;Y) = L_{g^{\ell}}^1(\epsilon).
	\end{align*}
\end{lemma}
\begin{proof}
	The proof follows since the kernel $P_{U|Y}^*$ that achieves the approximate solution satisfies the constraints in \eqref{main222}.
\end{proof}
Next we provide upper bounds for $g_{\epsilon}^{\ell}(P_{XY})$. To do so, we first bound the approximation error in \cite[Theorem~2]{Khodam22}. Let $\Omega^1$ be the set of all $\Omega_i\subset\{1,..,|\mathcal{Y}|\},\ |\Omega_i|=|\cal X|$, such that each $\Omega_i$ produces a valid standard distribution vector $M_{\Omega_i}^{-1}MP_Y$, i.e., all elements in the vector $M_{\Omega_i}^{-1}MP_Y$ are positive.
\begin{proposition}\label{mos}
	Let the approximation error be the distance between $H(Y|U)$ and the approximation derived in \cite[Theorem~2]{Khodam22}. Then, for all $\epsilon<\frac{1}{2}\epsilon_2$, we have 
	\begin{align*}
	|\text{Approximation\  error}|<\frac{3}{4}.
	\end{align*}
	Furthermore, for all $\epsilon<\frac{1}{2}\frac{\epsilon_2}{\sqrt{|\mathcal{X}|}}$ the upper bound can be strengthened as follows
	\begin{align*}
	|\text{Approximation\  error}|<\frac{1}{2(2\sqrt{|\mathcal{X}|}-1)^2}+\frac{1}{4|\mathcal{X}|}.
	\end{align*}
	where $\epsilon_2=\frac{\min_{y,\Omega\in \Omega^1} M_{\Omega}^{-1}MP_Y(y)}{\max_{\Omega\in \Omega^1} |\sigma_{\max} (H_{\Omega})|}$, $H_{\Omega}=M_{\Omega}^{-1}M(1:|\mathcal{X}|)P_{X|Y_1}^{-1}$ and $\sigma_{\max}$ is the largest right singular value. 
\end{proposition} 
\begin{proof}
	The proof is provided in Appendix~C.
\end{proof}
As a result we can find an upper bound on $g_{\epsilon}^{\ell}(P_{XY})$. To do so let $\text{approx}(g_{\epsilon}^{\ell})$ be the value that the kernel $P_{U^*|Y}$ in Lemma~\ref{lg2} achieves, i.e., the approximate value in \cite[(7)]{Khodam22}.
\begin{corollary}\label{ghahve}
	For any $0\leq\epsilon<\frac{1}{2}\epsilon_2$ we have
	\begin{align*}
	g_{\epsilon}^{\ell}(P_{XY})\leq \text{approx}(g_{\epsilon}^{\ell})+\frac{3}{4}=U_{g^{\ell}}^1(\epsilon),
	\end{align*}
	furthermore, for any $0\leq\epsilon<\frac{1}{2}\frac{\epsilon_2}{\sqrt{|\mathcal{X}|}}$ the upper bound can be strengthened as
	\begin{align*}
	g_{\epsilon}^{\ell}(P_{XY})\!\leq \text{approx}(g_{\epsilon}^{\ell})+\frac{1}{2(2\sqrt{|\mathcal{X}|}-1)^2}\!+\frac{1}{4|\mathcal{X}|}\!=\!U_{g^{\ell}}^2(\epsilon).
	\end{align*}
\end{corollary}
In the next theorem we summarize the bounds for $g_{\epsilon}^{\ell}(P_{XY})$ and $h_{\epsilon}^{\ell}(P_{XY})$, furthermore, a new upper bound for $h_{\epsilon}^{\ell}(P_{XY})$ is derived.
\begin{theorem}\label{koontala1}
	For any $0\leq\epsilon<\frac{1}{2}\epsilon_2$ and pair of RVs $(X,Y)$ distributed according to $P_{XY}$ supported on alphabets $\mathcal{X}$ and $\mathcal{Y}$ we have
	\begin{align*}
	L_{g^{\ell}}^1(\epsilon)\leq g_{\epsilon}^{\ell}(P_{XY})\leq U_{g^{\ell}}^1(\epsilon),
	\end{align*}
	and for any $0\leq\epsilon<\frac{1}{2}\frac{\epsilon_2}{\sqrt{|\mathcal{X}|}}$ we get
	\begin{align*}
	L_{g^{\ell}}^1(\epsilon)\leq g_{\epsilon}^{\ell}(P_{XY})\leq U_{g^{\ell}}^2(\epsilon),
	\end{align*}
	furthermore, for any $0\leq\epsilon$ 
	\begin{align*}
	g_{\epsilon}^{\ell}(P_{XY})\leq h_{\epsilon}^{\ell}(P_{XY})\leq \frac{\epsilon^2}{\min P_X}+H(Y|X)=U_{h^{\ell}}(\epsilon).
	\end{align*}
\end{theorem}
\begin{proof}
	It is sufficient to show that the upper bound on $h_{\epsilon}^{\ell}(P_{XY})$ holds, i.e., $U_{h^{\ell}}(\epsilon)$. To do so, let $U$ satisfy $d(P_{X|U}(\cdot|u),P_X)\leq\epsilon$, then we have
	\begin{align*}
	I(U;Y) &= I(X;U)\!+\!H(Y|X)\!-\!I(X;U|Y)\!-\!H(Y|X,U)\\ &\leq I(X;U)\!+\!H(Y|X)\\ &\stackrel{(a)}{\leq}\!\sum_u\!\! P_U(u)\frac{\left(d(P_{X|U}(\cdot|u),\!P_X)\right)^2}{\min P_X}\!+\!H(Y|X)  \\& = \frac{\epsilon^2}{\min P_X}+H(Y|X),
	\end{align*}
	where (a) follows by the reverse Pinsker inequality.
\end{proof}
In next section we study the special case where $X$ is a deterministic function of $Y$, i.e., $H(X|Y)=0$. 
\subsection*{Special case: $X$ is a deterministic function of $Y$}
In this case we have 
\begin{align}
h_{\epsilon}^{w\ell}(P_{XY})&=g_{\epsilon}^{w\ell}(P_{XY})\label{khatakar}\\&= \max_{\begin{array}{c} 
	\substack{P_{U|Y}:X-Y-U\\ \ d(P_{X,U}(\cdot,u),P_XP_{U}(u))\leq\epsilon,\ \forall u\\ |\mathcal{U}|\leq |\mathcal{Y}|}
	\end{array}}I(Y;U)\nonumber\\&= \sup_{\begin{array}{c} 
	\substack{P_{U|Y}: d(P_{X,U}(\cdot,u),P_XP_{U}(u))\leq\epsilon,\ \forall u\\ |\mathcal{U}|\leq |\mathcal{Y}|}
	\end{array}}I(Y;U),\nonumber\\
h_{\epsilon}^{\ell}(P_{XY})&=g_{\epsilon}^{\ell}(P_{XY})\label{khatakoon}\\&= \max_{\begin{array}{c} 
	\substack{P_{U|Y}:X-Y-U\\ \ d(P_{X|U}(\cdot|u),P_X)\leq\epsilon,\ \forall u\\ |\mathcal{U}|\leq |\mathcal{Y}|}
	\end{array}}I(Y;U)\nonumber\\&= \sup_{\begin{array}{c} 
	\substack{P_{U|Y}: d(P_{X|U}(\cdot|u),P_X)\leq\epsilon,\ \forall u\\ |\mathcal{U}|\leq |\mathcal{Y}|}
	\end{array}}I(Y;U),\nonumber
\end{align}
since the Markov chain $X-Y-U$ holds. Consequently, by using Theorem~2 and \eqref{khatakar} we have the next corollary.
\begin{corollary}\label{chooni}
	For any $0\leq \epsilon\leq \sqrt{2I(X;Y)}$ we have
	\begin{align*}
	\max\{L_{h^{w\ell}}^{1}(\epsilon),L_{h^{w\ell}}^{2}(\epsilon),L_{g^{w\ell}}^1(\epsilon)\}\leq g_{\epsilon}^{w\ell}(P_{XY})\leq U_{g^{w\ell}}(\epsilon).	
	\end{align*}
\end{corollary}
We can see that the bounds in Corollary~\ref{chooni} are asymptotically optimal. The latter follows since in the high privacy regime, i.e., the leakage tends to zero, $U_{g^{w\ell}}(\epsilon)$ and $L_{h^{w\ell}}^{1}(\epsilon)$ both tend to $H(Y|X)$, which is the optimal solution to $g_{0}(P_{XY})$ when $X$ is a deterministic function of $Y$, \cite[Theorem~6]{kostala}.
Furthermore, by using Theorem~3 and \eqref{khatakoon} we obtain the next result.
\begin{corollary}
	For any $0\leq\epsilon<\frac{1}{2}\epsilon_2$ we have
	\begin{align*}
	L_{g^{\ell}}^1(\epsilon)\leq g_{\epsilon}^{\ell}(P_{XY}) \leq \min\{U_{g^{\ell}}^1(\epsilon),U_{h^{\ell}}(\epsilon)\}.
	\end{align*}	
\end{corollary}
\begin{remark}
	For deriving the upper bound $U_{h^{\ell}}(\epsilon)$ and lower bounds $L_{h^{w\ell}}^1(\epsilon)$ and $L_{h^{w\ell}}^2(\epsilon)$ we do not use the assumption that the leakage matrix $P_{X|Y}$ is of full row rank. Thus, these bounds hold for all $P_{X|Y}$ and all $\epsilon\geq 0$.  
\end{remark}
Next result shows a property of the optimizers of $h_{\epsilon}^{w\ell}(P_{XY})$ and $h_{\epsilon}^{\ell}(P_{XY})$.
\begin{proposition}\label{kirr}
	Let $\bar{U}_1$ and $\bar{U}_2$ be any optimizers of $h_{\epsilon}^{w\ell}(P_{XY})$ and $h_{\epsilon}^{\ell}(P_{XY})$, respectively. Then we have
	\begin{align*}
	H(Y|X,\bar{U}_1)=H(Y|X,\bar{U}_2)=0.
	\end{align*}
\end{proposition}
\begin{proof}
	The proof follows similar arguments as for Lemma~\ref{kir}. In the proof of Lemma~\ref{kir}, instead of $\bar{U}$ use $\bar{U}_1$ and let $U'$ be produced in a similar way. The only difference is that instead of showing $I(U;X)\leq \epsilon$, we need to show that $d(P_{X,U}(\cdot,u),P_XP_U(u))\leq \epsilon$ for all $u$, where $U=(U',\bar{U}_1)$. The latter holds since $U'$ is independent of $(\bar{U}_1,X)$ and $\bar{U}_1$ satisfies the strong privacy criterion 1. The same proof works for $\bar{U}_2$.  
\end{proof}
In the next part, we study a numerical example to illustrate the new bounds.
\subsection*{Example}
Let us consider RVs $X$ and $Y$ with joint distribution $P_{XY}=\begin{bmatrix}
0.693 & 0.027 &0.108& 0.072\\0.006 & 0.085 & 0.004 & 0.005
\end{bmatrix}$. Using the definition of $\epsilon_2$ in Proposition~\ref{mos} we have $\epsilon_2 = 0.0341$. Fig.~\ref{kir12} illustrates the lower bound and upper bounds for $g_{\epsilon}^{\ell}$ derived in Theorem~\ref{koontala1}. As shown in Fig.~\ref{kir12}, the upper bounds $U_{g^{\ell}}^1(\epsilon)$ and $U_{g^{\ell}}^2(\epsilon)$ are valid for $\epsilon< 0.0171 $ and $\epsilon < 0.0121$, however the upper bound $U_{h^{\ell}}(\epsilon)$ is valid for all $\epsilon\geq 0$. In this example, we can see that for any $\epsilon$ the upper bound $U_{h^{\ell}}(\epsilon)$ is the smallest upper bound.
\begin{figure}[]
	\centering
	\includegraphics[scale = .135]{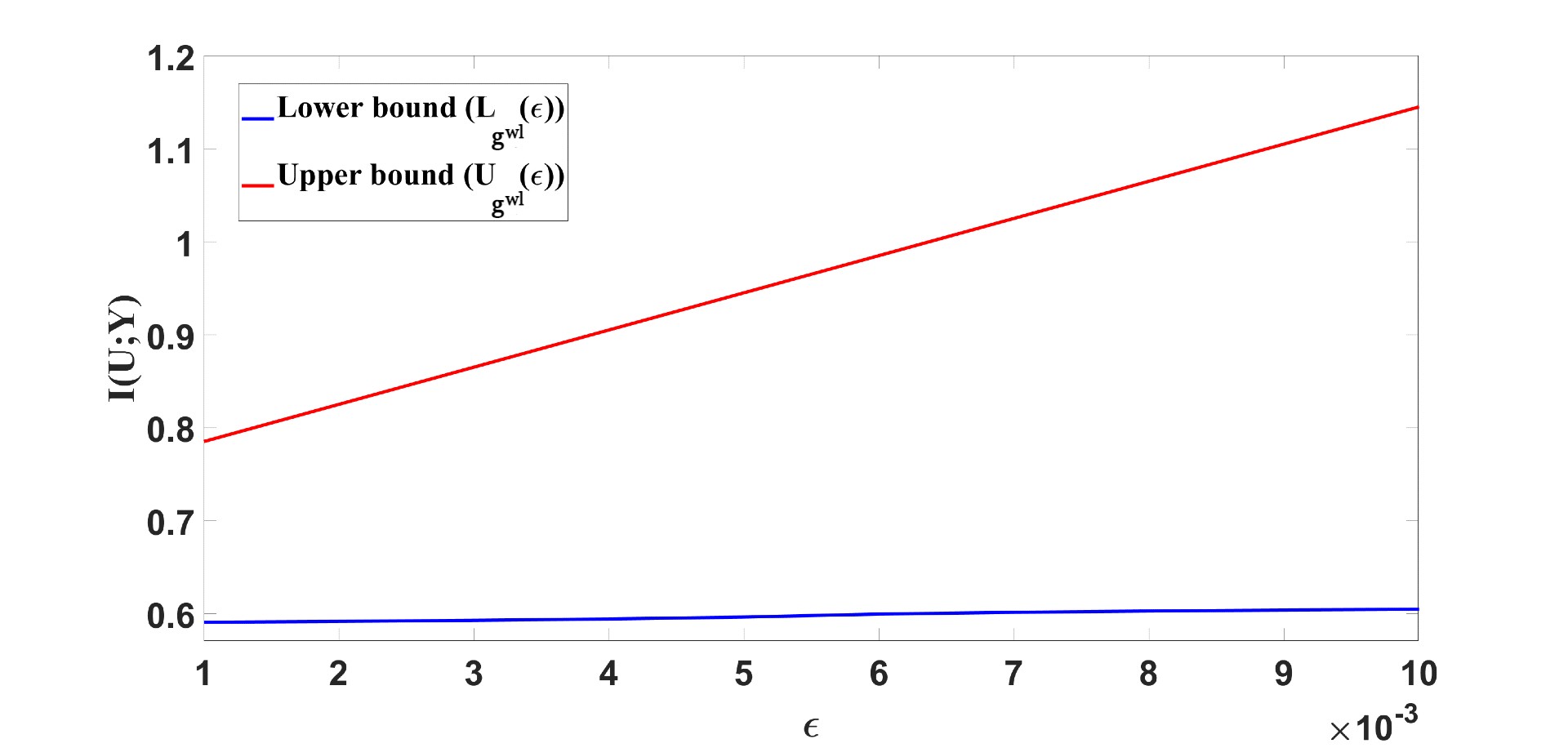}
	\caption{Comparing the upper bound and lower bound for $g_{\epsilon}^{w\ell}$.}
	\label{kir11}
\end{figure} 
Furthermore, Fig.~\ref{kir11} shows the lower bound $L_{g^{w\ell}}(\epsilon)$ and upper bound $U_{g^{w\ell}}(\epsilon)$ obtained in Theorem~\ref{choon1}. 
\begin{figure}[]
	\centering
	\includegraphics[scale = .135]{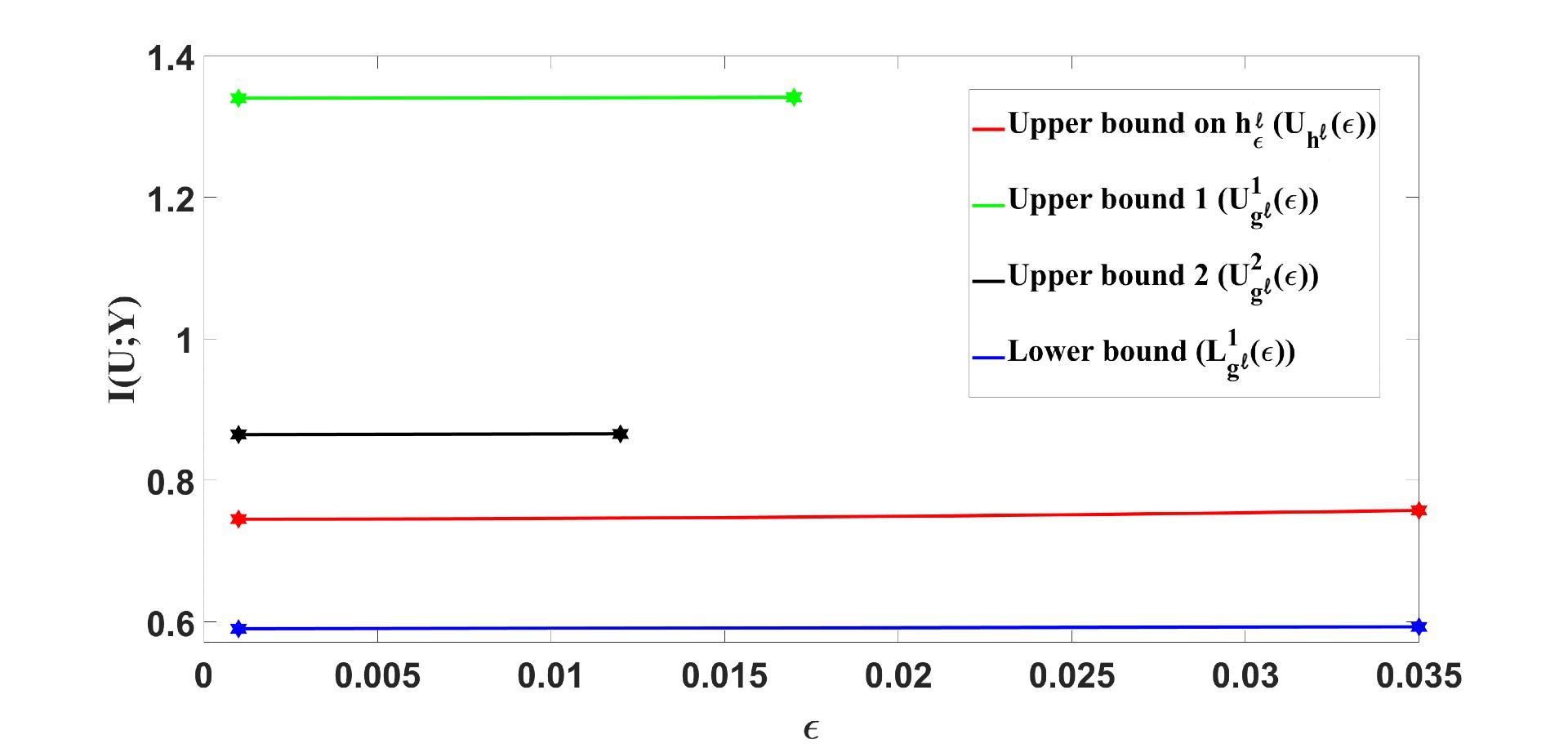}
	\caption{Comparing the upper bound and lower bound for $g_{\epsilon}^{\ell}$. The upper bounds $U_{g^{\ell}}^1(\epsilon)$ and $U_{g^{\ell}}^2(\epsilon)$ are valid for $\epsilon< 0.0171 $ and $\epsilon < 0.0121$, respectively. On the other hand, the upper bound $U_{h^{\ell}}(\epsilon)$ is valid for all $\epsilon\geq 0$.}
	\label{kir12}
\end{figure} 
\begin{figure}[]
	\centering
	\includegraphics[scale = .135]{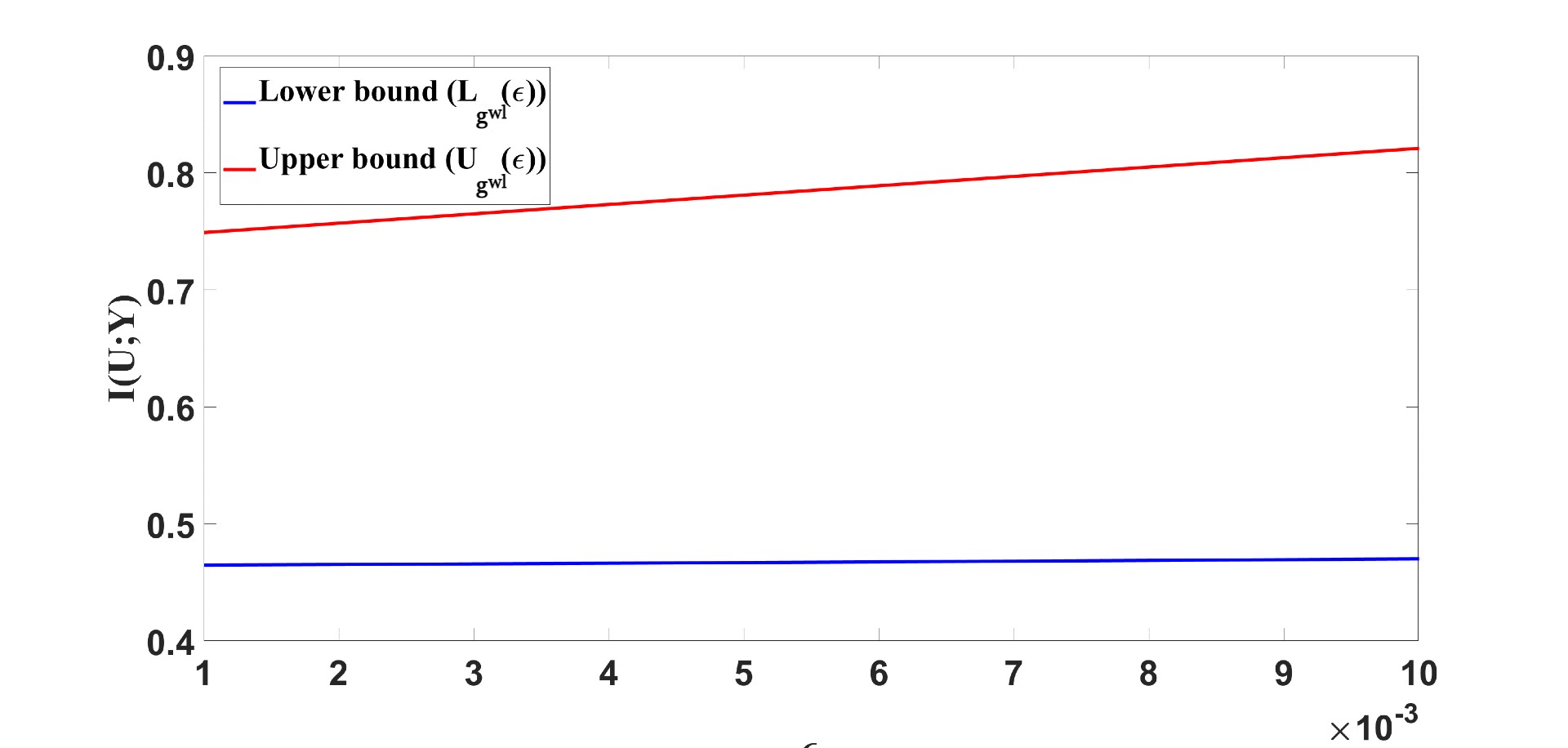}
	\caption{Comparing the upper bound and lower bound for $g_{\epsilon}^{w\ell}$.}
	\label{kir111}
\end{figure}  
\begin{figure}[]
	\centering
	\includegraphics[scale = .135]{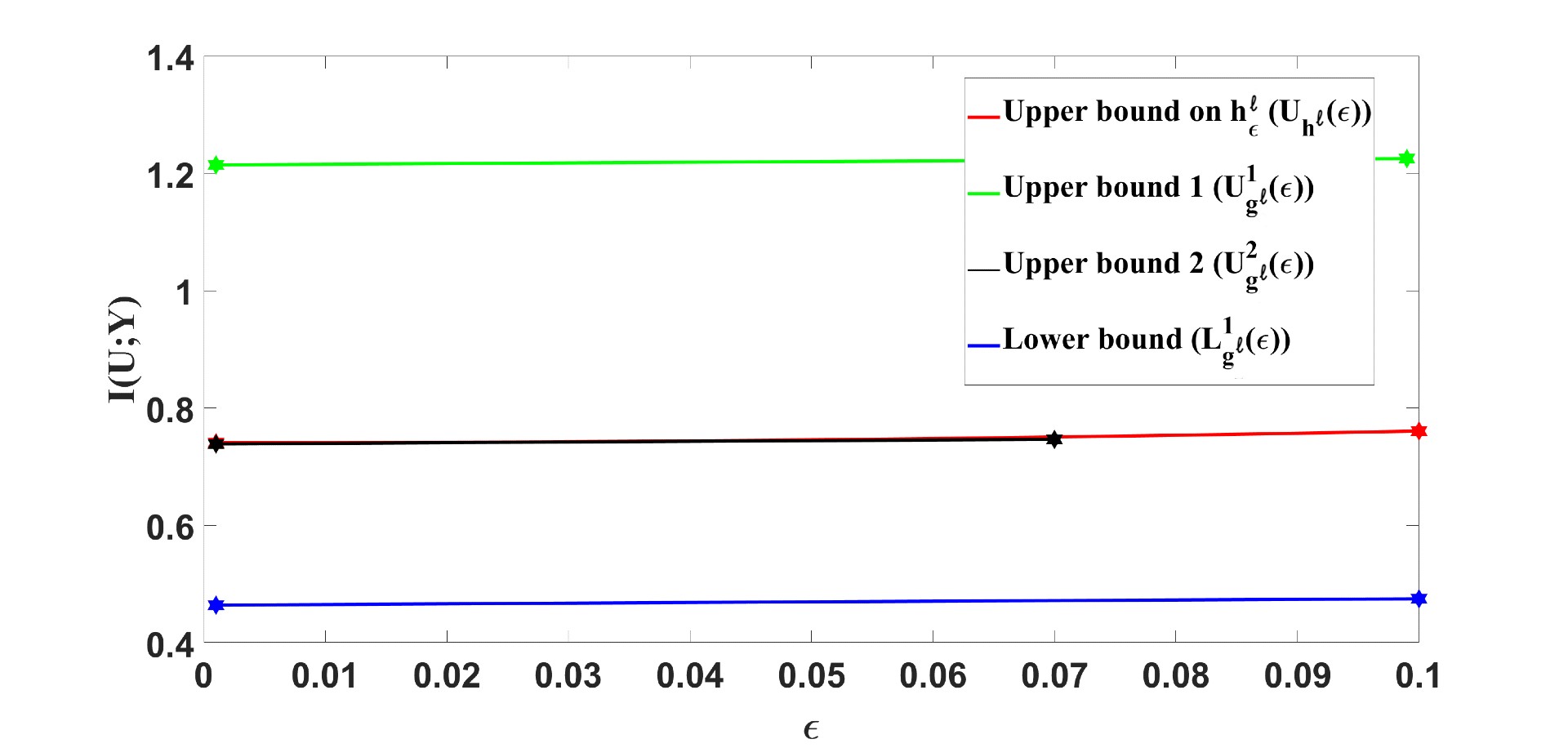}
	\caption{Comparing the upper bound and lower bound for $g_{\epsilon}^{\ell}$. The upper bounds $U_{g^{\ell}}^1(\epsilon)$ and $U_{g^{\ell}}^2(\epsilon)$ are valid for $\epsilon<0.0997 $ and $\epsilon <0.0705$, respectively. However, the upper bound $U_{h^{\ell}}(\epsilon)$ is valid for all $\epsilon\geq 0$.}
	\label{kir122}
\end{figure}  
Next, let $P_{XY}=\begin{bmatrix}
0.350 & 0.025 &0.085& 0.040\\0.025 & 0.425 & 0.035 & 0.015
\end{bmatrix}$. In this case, $\epsilon_2=0.1994$. Fig.~\ref{kir122} illustrates the lower bound and upper bounds for $g_{\epsilon}^{\ell}$. We can see that for $\epsilon<0.0705$, $U_{g^{\ell}}^2(\epsilon)$ is the smallest upper bound and for $\epsilon>0.0705$, $U_{h^{\ell}}(\epsilon)$ is the smallest bound. Furthermore, Fig.~\ref{kir111} shows the lower bound $L_{g^{w\ell}}(\epsilon)$ and upper bound $U_{g^{w\ell}}(\epsilon)$.
\subsection{Privacy-utility trade-off with non-zero leakage and prioritized private data}
In this part we find lower and upper bounds for $h_{\epsilon}^{p}(P_{X_1X_2Y})$ defined in \eqref{main111}. To find lower bounds we use similar techniques as used in Theorem~\ref{th.1} and Proposition~\ref{prop111}, i.e., we use extended versions of FRL and SFRL for correlated $(X_1,X_2)$ and $U$.
\begin{theorem}\label{koonbekoon}
	For any $0\leq \epsilon$ and RVs $(X_1,X_2,Y)$ distributed according to $P_{X_1X_2Y}$ supported on alphabets $\mathcal{X}_1$, $\mathcal{X}_2$ and $\mathcal{Y}$ we have
	\begin{align}\label{fuckbaba}
	\max\{L_{h^{p}}^{1}(\epsilon),L_{h^{p}}^{2}(\epsilon),L_{h^{p}}^{3}(\epsilon)\}\leq h_{\epsilon}^{p}(P_{X_1X_2Y})\leq U_{h^{p}}^{1}(\epsilon) ,
	\end{align}
	where
	\begin{align*}
	L_{h^{p}}^{1}(\epsilon) &= \epsilon+H(Y|X_1,X_2)-H(X_1,X_2|Y),\\
	L_{h^{p}}^{2}(\epsilon) &= \epsilon+H(Y|X_1,X_2)-\alpha H(X_2|Y) -\left( \log(I(X_1,X_2;Y)+1)+4 \right),\\
	L_{h^{p}}^{3}(\epsilon) &= \epsilon+H(Y|X_1,X_2)-\alpha H(X_1,X_2|Y) -(1-\alpha)\left( \log(I(X_1,X_2;Y)+1)+4 \right),\\
	U_{h^{p}}^{1}(\epsilon) &= \epsilon+H(Y|X_1,X_2),
	\end{align*}
	with $\alpha=\frac{\epsilon}{H(X_2)}$.
\end{theorem}
\begin{proof}
	The proof is provided in Appendix~D.
\end{proof}
To compare the lower bounds $L_{h^{p}}^{2}(\epsilon)$ and $L_{h^{p}}^{3}(\epsilon)$ we consider three cases as follows. First, let $X_1$ be a deterministic function of $Y$, then we have $H(X_1,X_2|Y)=H(X_2|Y)$. Hence, in this case $L_{h^{12}}^{1}(\epsilon)\geq L_{h^{p}}^{3}(\epsilon)\geq L_{h^{p}}^{2}(\epsilon)$. Next, let $X_2$ be a deterministic function of $Y$ and assume $4+H(Y)\leq H(X_1|Y)$. In this case, we have
\begin{align*}
L_{h^{p}}^{2}(\epsilon)-L_{h^{p}}^{3}(\epsilon)&=\alpha \left(H(X_1|Y)-\log(I(X_1;Y)+H(X_2|X_1)+1)-4\right)\\& \stackrel{(a)}{\geq} \alpha\left(H(X_1|Y)-I(X_1;Y)-H(X_2|X_1)-4\right)\\&  \stackrel{(b)}{\geq} \alpha\left(H(X_1|Y)-I(X_1;Y)-H(Y|X_1)-4\right)\\ &= \alpha\left( H(X_1|Y)-H(Y)-4\right)\\&\geq 0,
\end{align*} 
where (a) follows since $\log(1+x)\leq x$ and in step (b) we use $H(X_2|X_1)\leq H(Y|X_1)$ since $H(X_2|Y)=0$. So, in this case $L_{h^{p}}^{1}(\epsilon)\geq L_{h^{p}}^{2}(\epsilon)\geq L_{h^{p}}^{3}(\epsilon)$. Finally, let $Y$ be independent of $(X_1,X_2)$ and assume $H(X_1,X_2)\geq 4$. In this case we have $L_{h^{p}}^{2}(\epsilon)\geq L_{h^{p}}^{1}(\epsilon)$ and $L_{h^{p}}^{3}(\epsilon)\geq L_{h^{p}}^{1}(\epsilon)$.\\
The upper bound $U_{h^{p}}^{1}(\epsilon)$ is attained whenever the pair $(X_1,X_2)$ is a deterministic function of $Y$. In this case $U_{h^{p}}^{1}(\epsilon)=L_{h^{p}}^{1}(\epsilon)$.\\
Similar to Lemma~\ref{kir} and Proposition~\ref{kirr} it can be shown that if $\tilde{U}$ is an optimizer of $h_{\epsilon}^{p}(P_{X_1X_2Y})$, then $Y$ is a deterministic function of $\tilde{U}$ and $(X_1,X_2)$.
\begin{proposition}
	Let $\tilde{U}$ be an optimizer of $h_{\epsilon}^{p}(P_{X_1X_2Y})$, then
	\begin{align*}
	H(Y|X_1,X_2,\tilde{U})=0.
	\end{align*}
\end{proposition}
\begin{proof}
	The proof follows similar arguments as in Lemma~\ref{kir}. Let $\tilde{U}$ be an optimizer of $h_{\epsilon}^{p}(P_{X_1,X_2,Y})$ and $H(Y|X_1,X_2,\tilde{U})>0$. Consequently, $I(X_1,X_2;\tilde{U})\leq \epsilon$ and $I(X_1;\tilde{U})\leq I(X_2;\tilde{U})$. Let $U'$ be produced by FRL using $(X_1,X_2,\tilde{U})$ instead of $X$ in Lemma~\ref{lemma1} and same $Y$. Thus, $I(Y;U')>0$ and by letting $U=(U',\tilde{U})$ and using similar arguments as in Lemma~\ref{kir} we have $I(Y;U)>I(Y;\tilde{U})$. Furthermore,
	\begin{align*}
	I(X_1,X_2,U)&\stackrel{(a)}{=} I(X_1,X_2;\tilde{U})\leq \epsilon,\\
	I(X_1;U)&\stackrel{(b)}{=}I(X_1;\tilde{U})\leq I(X_2;\tilde{U})\stackrel{(c)}{=}I(X_2;U),
	\end{align*}
	where (a), (b) and (c) follow from the fact that $U'$ is independent of $(X_1,X_2,\tilde{U})$. Thus, $U$ achieves strictly larger utility than $\tilde{U}$ which contradicts the optimality of $\tilde{U}$. 
\end{proof}

\section{conclusion}\label{concul1}
Different information theoretic data disclosure problems have been studied in this work. The FRL and SRFL have been extended by relaxing the independence constraint and allowing certain amount of leakage using different privacy measures. It has been shown that by using extended versions of the FRL and SFRL lower bounds on privacy-utility trade-off functions can be derived. The results are useful since the proofs are constructive and therefore valuable for mechanism design and the bounds on optimality serve as a benchmark. 
Concepts from information geometry can be used to find lower bounds on privacy-utility trade-off functions considering first scenario when per-letter privacy constraints (strong privacy criterions) are used. 
\section*{Appendix A}
\subsection*{Proofs for Section~\ref{sec:system}:}
\textbf{\emph{Proof of Proposition \ref{gooz}:}} For each $u\in\mathcal{U}$ we have
\begin{align*}
\mathcal{L}^1(X;U=u)&= \left\lVert P_{X|U=u}(\cdot|u)-P_X \right\rVert_1\\&=\left\lVert P_{X|Y}(P_{Y|U=u}(\cdot|u)-P_Y)\right\rVert_1\\&=\sum_x |\sum_y P_{X|Y}(x,y)(P_{Y|U=u}(y)\!-\!P_Y(y))|\\
& \stackrel{(a)}{\leq} \sum_x\sum_y P_{X|Y}(x,y)|P_{Y|U=u}(y)-P_Y(y)|\\&=\sum_y\sum_x P_{X|Y}(x,y)|P_{Y|U=u}(y)-P_Y(y)|\\&=\sum_y |P_{Y|U=u}(y)-P_Y(y)|\\&= \left\lVert P_{Y|U=u}(\cdot|u)\!-\!P_Y \right\rVert_1=\mathcal{L}^1(Y;U=u),
\end{align*}
where (a) follows from the triangle inequality. Furthermore, we can multiply all the above expressions by the term $P_U(u)$ and we obtain
\begin{align*}
\mathcal{L}^2(X;U=u)\leq \mathcal{L}^2(Y;U=u).
\end{align*}
\section*{Appendix B}
\subsection*{Proofs for \emph{Privacy-utility trade-off with non-zero leakage}:}
\textbf{\emph{Proof of Lemma \ref{lemma3}:}}
Let $\tilde{U}$ be the RV found by FRL and let $W=\begin{cases}
X,\ \text{w.p}.\ \alpha\\
c,\ \ \text{w.p.}\ 1-\alpha
\end{cases}$, where $c$ is a constant which does not belong to the support of $X$ and $Y$ and $\alpha=\frac{\epsilon}{H(X)}$. We show that $U=(\tilde{U},W)$ satisfies the conditions. We have
\begin{align*}
I(X;U)&=I(X;\tilde{U},W)\\&=I(\tilde{U};X)+I(X;W|\tilde{U})\\&\stackrel{(a)}{=}H(X)-H(X|\tilde{U},W)\\&=H(X)-\alpha H(X|\tilde{U},X)-(1-\alpha)H(X|\tilde{U},c)\\&=H(X)-(1-\alpha)H(X)=\alpha H(X)=\epsilon,	
\end{align*}
where in (a) we used the fact that $X$ and $\tilde{U}$ are independent. Furthermore,
\begin{align*}
H(Y|X,U)&=H(Y|X,\tilde{U},W)\\&=\alpha H(Y|X,\tilde{U})+(1-\alpha)H(Y|X,\tilde{U},c)\\&=H(Y|X,\tilde{U})=0.
\end{align*}
In the last line we used the fact that $\tilde{U}$ is produced by FRL.\\
\textbf{\emph{Proof of Lemma \ref{lemma4}:}}
Let $\tilde{U}$ be the RV found by SFRL and $W$ be the same RV which is used to prove Lemma~\ref{lemma3}. It is sufficient to show that $I(X;U|Y)\leq \alpha H(X|Y)+(1-\alpha)\left[ \log(I(X;Y)+1)+4\right]$ since all other properties are already proved in Lemma~3. We have
\begin{align*}
I(X;\tilde{U},W|Y)&=I(X;\tilde{U}|Y)+I(X,W|\tilde{U},Y)\\&\stackrel{(a)}{=}I(X;\tilde{U}|Y)+\alpha H(X|\tilde{U},Y)\\&=I(X;\tilde{U}|Y)+\alpha(H(X|Y)-I(X;\tilde{U}|Y))\\&=\alpha H(X|Y)+(1-\alpha)I(X;\tilde{U}|Y)\\&\stackrel{(b)}{\leq} \!\alpha H(X|Y)\!+\!(1-\alpha)\!\left[ \log(I(X;Y)\!+\!1)\!+\!4\right],
\end{align*}
where in step (a) we used the fact that 
\begin{align*}
I(X,W|\tilde{U},Y) &= H(X|\tilde{U},Y)-H(X|W,\tilde{U},Y)\\&=
H(X|\tilde{U},Y)-(1-\alpha)H(X|\tilde{U},Y)\\&=\alpha H(X|\tilde{U},Y),
\end{align*}
and (b) follows since $\tilde{U}$ is produced by SFRL.
\begin{lemma}\label{haa}
	For any pair of RVs $(X,Y)$ distributed according to $P_{XY}$ supported on alphabets $\mathcal{X}$ and $\mathcal{Y}$, where $|\mathcal{X}|$ is finite and $|\mathcal{Y}|$ is finite or countably infinite, there exists RV $U$ such that it satisfies \eqref{c1}, \eqref{c2}, and
	\begin{align*}
	H(U)\leq \sum_{x\in\mathcal{X}}H(Y|X=x)+\epsilon+h(\alpha)
	\end{align*}
	with $\alpha=\frac{\epsilon}{H(X)}$ and $h(\cdot)$ denotes the binary entropy function.
\end{lemma}
\begin{proof}
	Let $U=(\tilde{U},W)$ where $W$ is the same RV used in Lemma~\ref{lemma3} and $\tilde{U}$ is produced by FRL which has the same construction as used in proof of \cite[Lemma~1]{kostala}. Thus, by using \cite[Lemma~2]{kostala} we have
	\begin{align*}
	H(\tilde{U})\leq \sum_{x\in\mathcal{X}} H(Y|X=x),
	\end{align*} 
	therefore,
	\begin{align*}
	H(U)&=H(\tilde{U},W)\leq H(\tilde{U})+H(W),\\&\leq\sum_{x\in\mathcal{X}} H(Y|X=x)+H(W),
	\end{align*}
	where,
	\begin{align*}
	H(W)\! &= -(1-\alpha)\log(1-\alpha)\!-\!\!\sum_{x\in \mathcal{X}} \alpha P_X(x)\log(\alpha P_X(x)),\\&=h(\alpha)+\alpha H(X),
	\end{align*}
	which completes the proof.
\end{proof}
\textbf{\emph{Proof of Theorem \ref{th.1}:}} $L_3^{\epsilon}$ can be derived by using \cite[Remark~2]{shahab}, since we have $h_{\epsilon}(P_{XY})\geq g_{\epsilon}(P_{XY})\geq L_3^{\epsilon}$. For deriving $L_1$, let $U$ be produced by EFRL. Thus, using the construction of $U$ as in Lemma~\ref{lemma3} we have $I(X,U)=\epsilon$ and $H(Y|X,U)=0$. Then, using \eqref{key} we obtain
\begin{align*}
h_{\epsilon}(P_{XY})&\geq I(U;Y)\\&=I(X;U)\!+\!H(Y|X)\!-\!H(Y|U,X)\!-\!I(X;U|Y)\\&=\epsilon+H(Y|X)-H(X|Y)+H(X|Y,U)\\ &\geq\epsilon+H(Y|X)-H(X|Y)=L_1.
\end{align*} For deriving $L_2^{\epsilon}$, let $U$ be produced by ESFRL. Thus, by using the construction of $U$ as in Lemma~\ref{lemma4} we have $I(X,U)=\epsilon$, $H(Y|X,U)=0$ and $I(X;U|Y)\leq \alpha H(X|Y)+(1-\alpha)\left(\log(I(X;Y)+1)+4\right)$. Then, by using \eqref{key} we obtain
\begin{align*}
h_{\epsilon}(P_{XY})&\geq I(U;Y)\\&=I(X;U)\!+\!H(Y|X)\!-\!H(Y|U,X)\!-\!I(X;U|Y)\\&=\epsilon+H(Y|X)-I(X;U|Y)\\&\geq\epsilon+H(Y|X)-\alpha H(X|Y)\\&\ +(1-\alpha)\left(\log(I(X;Y)+1)+4\right)=L_2^{\epsilon}.
\end{align*}
Let $X$ be a deterministic function of $Y$. In this case, set $\epsilon=0$ in $L_1^{\epsilon}$ so that we obtain $h_0(P_{XY})\geq H(Y|X)$. Furthermore, by using \eqref{key} we have $h_0(P_{XY})\leq H(Y|X)$. Moreover, since $X$ is a deterministic function of $Y$, the Markov chain $X-Y-U$ holds and we have $h_0(P_{XY})=g_0(P_{XY})=H(Y|X)$. Therefore, $L_3^{\epsilon}$ can be rewritten as
\begin{align*}
L_3^{\epsilon}&=\epsilon\frac{H(Y)}{H(X)}+H(Y|X)\left(\frac{H(X)-\epsilon}{H(X)}\right),\\&=\epsilon\frac{H(Y)}{H(X)}+(H(Y)-H(X))\left(\frac{H(X)-\epsilon}{H(X)}\right),\\&=H(Y)-H(X)+\epsilon.
\end{align*}
$L_2^{\epsilon}$ can be rewritten as follows
\begin{align*}
L_2^{\epsilon}=H(Y|X)+\epsilon-(1-\frac{\epsilon}{H(X)})(\log(H(X)+1)+4).
\end{align*}
Thus, if $H(X|Y)=0$, then $L_1^{\epsilon}=L_3^{\epsilon}\geq L_2^{\epsilon}$. Now we show that $L_1^{\epsilon}=L_3^{\epsilon}$ is tight. By using \eqref{key} we have
\begin{align*}
I(U;Y) &\stackrel{(a)}{=} I(X;U)+H(Y|X)-H(Y|U,X),\\&\leq \epsilon+H(Y|X)=L_1^{\epsilon}=L_3^{\epsilon}.
\end{align*} 
where (a) follows since $X$ is deterministic function of $Y$ which leads to $I(X;U|Y)=0$. Thus, if $H(X|Y)=0$, the lower bound in \eqref{th2} is tight.
Now suppose that the lower bound $L_1^{\epsilon}$ is tight and $X$ is not a deterministic function of $Y$. Let $\tilde{U}$ be produced by FRL using the construction of \cite[Lemma~1]{kostala}. As argued in the proof of \cite[Th.~6]{kostala}, there exists $x\in\cal X$ and $y_1,y_2\in\cal Y$ such that $P_{X|\tilde{U},Y}(x|\tilde{u},y_1)>0$ and $P_{X|\tilde{U},Y}(x|\tilde{u},y_2)>0$ which results in $H(X|Y,\tilde{U})>0$. Let $U=(\tilde{U},W)$ where $W$ is defined in Lemma~\ref{lemma3}. For such $U$ we have
\begin{align*}
H(X|Y,U)&=(1-\alpha)H(X|Y,\tilde{U})>0,\\
\Rightarrow I(U;Y)&\stackrel{(a)}{=}\epsilon+H(Y|X)-H(X|Y)+H(X|Y,U)\\
&>\epsilon+H(Y|X)-H(X|Y).
\end{align*}
where in (a) we used the fact that such $U$ satisfies $I(X;U)=\epsilon$ and $H(Y|X,U)=0$. The last line is a contradiction with tightness of $L_1^{\epsilon}$, since we can achieve larger values, thus, $X$ needs to be a deterministic function of $Y$.
\begin{lemma}\label{koonimooni}
	For any pair of RVs $(X,Y)$ distributed according to $P_{XY}$ supported on alphabets $\mathcal{X}$ and $\mathcal{Y}$, then if $U$ satisfies $I(X;U)\leq \epsilon$, $H(Y|X,U)=0$ and $|\mathcal{U}|\leq \left[|\mathcal{X}|(|\mathcal{Y}|-1)+1\right]\left[|\mathcal{X}|+1\right]$, we have
	\begin{align*}
	\sup_U H(U)\!&\geq\! \alpha H(Y|X)\!+\!(1-\alpha)(\max_{x\in\mathcal{X}}H(Y|X=x))\!\\&+\!h(\alpha)\!+\!\epsilon \geq H(Y|X)\!+\!h(\alpha)\!+\!\epsilon,
	\end{align*}
	where $\alpha=\frac{\epsilon}{H(X)}$ and $h(\cdot)$ corresponds to the binary entropy.
\end{lemma} 
\begin{proof}
	Let $U=(\tilde{U},W)$ where $W=\begin{cases}
	X,\ \text{w.p}.\ \alpha\\
	c,\ \ \text{w.p.}\ 1-\alpha
	\end{cases}$, and $c$ is a constant which does not belong to the support of $X$, $Y$ and $\tilde{U}$, furthermore, $\tilde{U}$ is produced by FRL. Using \eqref{key} and \cite[Lemma~3]{kostala} we have
	\begin{align}
	H(\tilde{U}|Y)&=H(\tilde{U})-H(Y|X)+I(X;\tilde{U}|Y)\nonumber\\ &\stackrel{(a)}{\geq} \max_{x\in\mathcal{X}} H(Y|X=x)-H(Y|X)\nonumber\\&\ +H(X|Y)-H(X|Y,\tilde{U}),\label{toole}
	\end{align} 	
	where (a) follows from \cite[Lemma~3]{kostala}. Furthermore, in the first line we used $I(X;\tilde{U})=0$ and $H(Y|\tilde{U},X)=0$. Using \eqref{key} we obtain
	\begin{align*}
	H(U)&\stackrel{(a)}{=}\!H(U|Y)\!+\!H(Y|X)\!-\!H(X|Y)\!+\!\epsilon\!+\!H(X|Y,U),\\
	&\stackrel{(b)}{=}H(W|Y)+\alpha H(\tilde{U}|Y,X)+(1-\alpha)H(\tilde{U}|Y) +H(Y|X)-H(X|Y)\!+\!\epsilon+(1-\alpha)H(X|Y,\tilde{U}),\\
	&\stackrel{(c)}{=}(\alpha-1)H(X|Y)+h(\alpha)+\alpha H(\tilde{U}|Y,X)+\epsilon +(1-\alpha)H(\tilde{U}|Y)+H(Y|X)+(1-\alpha)H(X|Y,\tilde{U}), \\
	& \stackrel{(d)}{\geq}  (\alpha-1)H(X|Y)+h(\alpha)+\alpha H(\tilde{U}|Y,X)+(1-\alpha) (\max_{x\in\mathcal{X}}H(Y|X=x)-H(Y|X)+H(X|Y) \\& -H(X|Y,\tilde{U}))+H(Y|X)+\epsilon+(1-\alpha)H(X|Y,\tilde{U})\\
	&=\alpha H(Y|X)+(1-\alpha)(\max_{x\in\mathcal{X}}H(Y|X=x))+h(\alpha)+\epsilon.
	\end{align*}
	In step (a) we used $I(U;X)=\epsilon$ and $H(Y|X,U)=0$ and in step (b) we used $H(U|Y)=H(W|Y)+H(\tilde{U}|Y,W)=H(W|Y)+\alpha H(\tilde{U}|Y,X)+(1-\alpha)H(\tilde{U}|Y)$ and $H(X|Y,U)=H(X|Y,\tilde{U},W)=(1-\alpha)H(X|Y,\tilde{U})$. In step (c) we used the fact that $P_{W|Y}=\begin{cases}
	\alpha P_{X|Y}(x|\cdot)\ &\text{if}\ w=x,\\
	1-\alpha \ &\text{if} \ w=c,
	\end{cases}$
	since $P_{W|Y}(w=x|\cdot)=\frac{P_{W,Y}(w=x,\cdot)}{P_Y(\cdot)}=\frac{P_{Y|W}(\cdot|w=x)P_W(w=x)}{P_Y(\cdot)}=\frac{P_{Y|X}(\cdot|x)\alpha P_X(x)}{P_Y(\cdot)}=\alpha P_{X|Y}(x|\cdot)$, furthermore, $P_{W|Y}(w=c|\cdot)=1-\alpha$. Hence, after some calculation we obtain $H(W|Y)=h(\alpha)+\alpha H(X|Y)$. Finally, step (d) follows from \eqref{toole}.
\end{proof}
\begin{remark}
	The constraint $|\mathcal{U}|\leq \left[|\mathcal{X}|(|\mathcal{Y}|-1)+1\right]\left[|\mathcal{X}|+1\right]$ in Lemma~\ref{koonimooni} guarantees that $\sup_U H(U)<\infty$.
\end{remark}
\textbf{\emph{Proof of Theorem~\ref{ziba}:}}
\begin{itemize}
	\item i $\Rightarrow$ ii: Using Lemma~\ref{goh} we have $H(Y|X)+\epsilon= g_{\epsilon}(P_{XY})\leq h_{\epsilon}(P_{XY}) \leq H(Y|X)+\epsilon$. Thus, $g_{\epsilon}(P_{XY})=h_{\epsilon}(P_{XY})$.
	\item ii $\Rightarrow$ iii: Let $\bar{U}$ be an optimizer of $g_{\epsilon}(P_{XY})$. Thus, the Markov chain $X-Y-\bar{U}$ holds and we have $I(X;\bar{U}|Y)=0$. Furthermore, since $g_{\epsilon}(P_{XY})=h_{\epsilon}(P_{XY})$ this $\bar{U}$ achieves $h_{\epsilon}(P_{XY})$. Thus, by using Lemma~\ref{kir} we have $H(Y|\bar{U},X)=0$ and according to \eqref{key} 
	\begin{align}
	I(\bar{U};Y)&=I(X;\bar{U})\!+\!H(Y|X)\!-\!H(Y|\bar{U},X)-I(X;\bar{U}|Y)\nonumber\\
	&=I(X;\bar{U})\!+\!H(Y|X).\label{kir2}
	\end{align}
	We claim that $\bar{U}$ must satisfy $I(X;Y|\bar{U})>0$ and $I(X;\bar{U})=\epsilon$. For the first claim assume that $I(X;Y|\bar{U})=0$, hence the Markov chain $X-\bar{U}-Y$ holds. Using $X-\bar{U}-Y$ and $H(Y|\bar{U},X)=0$ we have $H(Y|\bar{U})=0$, hence $Y$ and $\bar{U}$ become independent. Using \eqref{kir2}
	\begin{align*}
	H(Y)&=I(Y;\bar{U})=I(X;\bar{U})\!+\!H(Y|X),\\
	&\Rightarrow I(X;\bar{U})=I(X;Y).
	\end{align*}
	The last line is a contradiction since by assumption we have $I(X;\bar{U})\leq \epsilon < I(X;Y)$. Thus, $I(X;Y|\bar{U})>0$.
	For proving the second claim assume that $I(X;\bar{U})=\epsilon_1<\epsilon$. Let $U=(\bar{U},W)$ where $W=\begin{cases}
	Y,\ \text{w.p}.\ \alpha\\
	c,\ \ \text{w.p.}\ 1-\alpha
	\end{cases}$,
	and $c$ is a constant that $c\notin \mathcal{X}\cup \mathcal{Y}\cup \mathcal{\bar{U}}$ and $\alpha=\frac{\epsilon-\epsilon_1}{I(X;Y|\bar{U})}$.
	We show that $\frac{\epsilon-\epsilon_1}{I(X;Y|\bar{U})}<1$. By the assumption we have
	\begin{align*}
	\frac{\epsilon-\epsilon_1}{I(X;Y|\bar{U})}<\frac{I(X;Y)-I(X;\bar{U})}{I(X;Y|\bar{U})}\stackrel{(a)}{\leq}1,
	\end{align*}
	where step (a) follows since $I(X;Y)-I(X;\bar{U})-I(X;Y|\bar{U})=I(X;Y)-I(X;Y,\bar{U})\leq 0$. It can be seen that such $U$ satisfies $H(Y|X,U)=0$ and $I(X;U|Y)=0$ since
	\begin{align*}
	H(Y|X,U)&=\alpha H(Y|X,\bar{U},Y)+ (1-\alpha)H(Y|X,\bar{U})=0,\\
	I(X;U|Y)&= H(X|Y)-H(X|Y,\bar{U},W)\\&=H(X|Y)-\alpha H(X|Y,\bar{U})-(1-\alpha) H(X|Y,\bar{U})\\&=H(X|Y)-H(X|Y)=0, 
	\end{align*}
	where in deriving the last line we used the Markov chain $X-Y-\bar{U}$. Furthermore, 
	\begin{align*}
	I(X;U)&=I(X;\bar{U},W)=I(X;\bar{U})+I(X;W|\bar{U})\\
	&=I(X;\bar{U})+\alpha H(X|\bar{U})-\alpha H(X|\bar{U},Y)\\
	&=I(X;\bar{U})+\alpha I(X;Y|\bar{U})\\
	&=\epsilon_1+\epsilon-\epsilon_1=\epsilon,
	\end{align*}
	and
	\begin{align*}
	I(Y;U)&=I(X;U)\!+\!H(Y|X)\!-\!H(Y|U,X)-I(X;U|Y)\\
	&=\epsilon+H(Y|X).
	\end{align*}
	Thus, if $I(X;\bar{U})=\epsilon_1<\epsilon$ we can substitute $\bar{U}$ by $U$ for which $I(U;Y)>I(\bar{U};Y)$. This is a contraction and we conclude that $I(X;\bar{U})=\epsilon$ which proves the second claim. Hence, \eqref{kir2} can be rewritten as
	\begin{align*}
	I(\bar{U};Y)=\epsilon+H(Y|X).
	\end{align*}
	As a result $h_\epsilon(P_{XY})=\epsilon+H(Y|X)$ and the proof is completed.
	\item iii $\Rightarrow$ i: Let $\bar{U}$ be the optimizer of $h_{\epsilon}(P_{XY})$ and $h_{\epsilon}(P_{XY})=H(Y|X)+\epsilon$. Using Lemma~\ref{kir} we have $H(Y|\bar{U},X)=0$. By using \eqref{key} we must have 
	$I(X;\bar{U}|Y)=0$ and $I(X;\bar{U})=\epsilon$. We conclude that for this $\bar{U}$, the Markov chain $X-Y-\bar{U}$ holds and as a result $\bar{U}$ achieves $g_{\epsilon}(P_{XY})$ and we have $g_{\epsilon}(P_{XY})=H(Y|X)+\epsilon$. 
\end{itemize}
\textbf{\emph{Proof of Lemma~\ref{kir}:}} Let $\bar{U}$ be an optimizer of $h_{\epsilon}(P_{XY})$ and assume that $H(Y|X,\bar{U})>0$. Consequently, we have $I(X;\bar{U})\leq \epsilon.$ Let $U'$ be founded by FRL  with $(X,\bar{U})$ instead of $X$ in Lemma~\ref{lemma1} and same $Y$, that is $I(U';X,\bar{U})=0$ and $H(Y|X,\bar{U},U')=0$. Using \cite[Th.~5]{kostala} we have
\begin{align*}
I(Y;U')>0,
\end{align*}
since we assumed $H(Y|X,\bar{U})>0$. Let $U=(\bar{U},U')$ and we first show that $U$ satisfies $I(X;U)\leq \epsilon$. We have
\begin{align*}
I(X;U)&=I(X;\bar{U},U')=I(X;\bar{U})+I(X;U'|\bar{U}),\\
&=I(X;\bar{U})+H(U'|\bar{U})-H(U'|\bar{U},X),\\
&=I(X;\bar{U})+H(U')-H(U')\leq \epsilon,
\end{align*}
where in last line we used the fact that $U'$ is independent of the pair $(X,\bar{U})$. Finally, we show that $I(Y;U)>I(Y,\bar{U})$ which is a contradiction with optimality of $\bar{U}$. We have
\begin{align*}
I(Y;U)&=I(Y;\bar{U},U')=I(Y;U')+I(Y;\bar{U}|U'),\\
&=I(Y;U')+I(Y,U';\bar{U})-I(U';\bar{U})\\
&= I(Y;U')+I(Y,\bar{U})+I(U';\bar{U}|Y)-I(U';\bar{U})\\
&\stackrel{(a)}{\geq} I(Y;U')+I(Y,\bar{U})\\
&\stackrel{(b)}{>} I(Y,\bar{U}),
\end{align*}
where in (a) follows since $I(U';\bar{U}|Y)\geq 0$ and $I(U';\bar{U})=0$. Step (b) follows since $I(Y;U')>0$. Thus, the obtained contradiction completes the proof.\\
\textbf{\emph{Proof of Lemma~\ref{ankhar}:}} Since $X$ is a deterministic function of $Y$, for any $y\in \cal Y$ we have 
\begin{align*}
P_{Y|X}(y|x)=\begin{cases}
\frac{P_Y(y)}{P_X(x)},\ &x=f(y)\\
0, \ &\text{else} 
\end{cases},
\end{align*}
thus,
\begin{align*}
\sum_{y\in\mathcal{Y}}\int_{0}^{1} \mathbb{P}_X\{P_{Y|X}(y|X)\geq t\}\log (\mathbb{P}_X\{P_{Y|X}(y|X)\geq t\})dt
&=	\sum_{y\in\mathcal{Y}}\!\int_{0}^{\frac{P_Y(y)}{P_X(x=f(y))}} \!\!\!\!\!\!\!\!\!\!\!\!\!\!\!\!\!\!\!\!\!\!\!\!\!\mathbb{P}_X\{P_{Y|X}(y|X)\geq t\}\log (\mathbb{P}_X\{P_{Y|X}(y|X)\geq t\})dt\\
&= \sum_{y\in\mathcal{Y}} \frac{P_Y(y)}{\mathbb{P}_X\{x=f(y)\}}\mathbb{P}_X\{x=f(y)\}\log(\mathbb{P}_X\{x=f(y)\})\\
&= \sum_{y\in\mathcal{Y}} P_Y(y)\log(\mathbb{P}_X\{x=f(y)\})\\
&= \sum_{y\in\mathcal{Y}} P_X(x)\log(P_X(x))=-H(X)=-I(X;Y),
\end{align*}
where in last line we used \\ $\sum_{y\in\mathcal{Y}} P_Y(y)\log(\mathbb{P}_X\{x=f(y)\})=\sum_{x\in\mathcal{X}} \sum_{y:x=f(y)} P_Y(y)\log(\mathbb{P}_X\{x=f(y)\})=\sum_{x\in\mathcal{X}} P_X(x)\log(P_X(x))$.
\section*{Appendix C}
\subsection*{Proofs for \emph{Privacy-utility trade-off with non-zero leakage and per-letter privacy constraints}:}
\textbf{\emph{Proof of Theorem~\ref{choon1}:}}
Lower bounds on $ g_{\epsilon}^{w\ell}(P_{XY})$ and $h_{\epsilon}^{w\ell}(P_{XY})$ are derived in Lemma~\ref{lg1} and Proposition~\ref{prop111}, respectively. Furthermore, inequality $g_{\epsilon}^{w\ell}(P_{XY})\leq h_{\epsilon}^{w\ell}(P_{XY})$ holds since $h_{\epsilon}^{w\ell}(P_{XY})$ has less constraints. To prove the upper bound on $g_{\epsilon}^{w\ell}(P_{XY})$, i.e., $U_{g^{w\ell}}(\epsilon)$, let $U$ satisfy $X-Y-U$ and $d(P_{X,U}(\cdot,u),P_XP_U(u))\leq\epsilon$, then we have
\begin{align*}
I(U;Y) &= I(X;U)\!+\!H(Y|X)\!-\!I(X;U|Y)\!-\!H(Y|X,U)\\
&\stackrel{(a)}{=} I(X;U)\!+\!H(Y|X)-H(Y|X,U)\\&\leq I(X;U)\!+\!H(Y|X)\\&=\sum_u P_U(u)D(P_{X|U}(\cdot|u),P_X)+H(Y|X)\\ &\stackrel{(b)}{\leq}\!\sum_u\!\! P_U(u)\frac{\left(d(P_{X|U}(\cdot|u),\!P_X)\right)^2}{\min P_X}\!+\!H(Y|X) \\&\stackrel{(c)}{\leq}\!\sum_u\!\! P_U(u)\frac{d(P_{X|U}(\cdot|u),\!P_X)}{\min P_X}|\mathcal{X}|\!+\!H(Y|X) \\&=\sum_u \!\!\frac{d(P_{X|U}(\cdot,u),\!P_XP_U(u))}{\min P_X}|\mathcal{X}|\!+\!H(Y|X)\\&\stackrel{(d)}{\leq} \frac{\epsilon|\mathcal{Y}||\mathcal{X}|}{\min P_X}+H(Y|X), 
\end{align*} 
where (a) follows by the Markov chain $X-Y-U$, (b) follows by the reverse Pinsker inequality \cite[(23)]{verdu} and (c) holds since $d(P_{X|U}(\cdot|u),\!P_X)=\sum_{i=1}^{|\mathcal{X}|} |P_{X|U}(x_i|u)-P_X(x_i)|\leq |\mathcal{X}|$. Latter holds since for each $u$ and $i$, $|P_{X|U}(x_i|u)-P_X|\leq 1$. Moreover, (d) holds since by Proposition~\ref{prop222} without loss of optimality we can assume $|\mathcal{U}|\leq |\mathcal{Y}|$. In other words (d) holds since by Proposition~\ref{prop222} we have
\begin{align}
g_{\epsilon}^{w\ell}(P_{XY})&=\sup_{\begin{array}{c} 
	\substack{P_{U|Y}:X-Y-U\\ \ d(P_{X,U}(\cdot,u),P_XP_{U}(u))\leq\epsilon,\ \forall u}
	\end{array}}I(Y;U)\nonumber\\&= \max_{\begin{array}{c} 
	\substack{P_{U|Y}:X-Y-U\\ \ d(P_{X,U}(\cdot,u),P_XP_{U}(u))\leq\epsilon,\ \forall u\\ |\mathcal{U}|\leq |\mathcal{Y}|}
	\end{array}}I(Y;U).\label{koontala}
\end{align}
\textbf{\emph{Proof of Proposition~\ref{mos}:}} By using \cite[Proposition~2]{Khodam22}, it suffices to assume $|\mathcal{U}|\leq|\mathcal{Y}|$. Using \cite[Proposition~3]{Khodam22}, let us consider $|\mathcal{Y}|$ extreme points that achieves the minimum in \cite[Theorem~2]{Khodam22} as $V_{\Omega_j}$ for $j\in\{1,..,|\mathcal{Y}|\}$. 
Let $|\mathcal{X}|$ non-zero elements of $V_{\Omega_j}$ be $a_{ij}+\epsilon b_{ij}$ for  $i\in\{1,..,|\mathcal{X}|\}$ and $j\in\{1,..,|\mathcal{Y}|\}$, where $a_{ij}$ and $b_{ij}$ can be found in \cite[(6)]{Khodam22}. 
As a summary for $i\in\{1,..,|\mathcal{X}|\}$ and $j\in\{1,..,|\mathcal{Y}|\}$ we have $\sum_i a_{ij}=1$, $\sum_i b_{ij}=0$, $0\leq a_{ij}\leq 1$, and $0\leq a_{ij}+\epsilon b_{ij}\leq1.$
We obtain
\begin{align*}
\max I(U;Y)&=H(Y)\sum_jP_j\sum_i (a_{ij}+\epsilon b_{ij})\log(a_{ij}+\epsilon b_{ij}),\\
&=H(Y)+\sum_jP_j\sum_i (a_{ij}+\epsilon b_{ij})(\log(a_{ij})+\log(1+\epsilon\frac{b_{ij}}{a_{ij}})).
\end{align*} 
In \cite[Theorem~2]{Khodam22}, we have used the Taylor expansion to derive the approximation of the equivalent problem. From the Taylor's expansion formula we have 
\begin{align*}
f(x)&=f(a)+\frac{f'(a)}{1!}(x-a)+\frac{f''(a)}{2!}(x-a)^2+...+\frac{f^{(n)}(a)}{n!}(x-a)^n+R_{n+1}(x),
\end{align*} 
where
\begin{align}\label{kosss}
R_{n+1}(x)&=\int_{a}^{x}\frac{(x-t)^n}{n!}f^{(n+1)}(t)dt\\&=\frac{f^{(n+1)}(\zeta)}{(n+1)!}(x-a)^{n+1},
\end{align}
for some $\zeta\in[a,x]$. 
In \cite{Khodam22} we approximated the terms $\log(1+\frac{b_{ij}}{a_{ij}}\epsilon)$ by $\frac{b_{ij}}{a_{ij}}\epsilon+o(\epsilon)$. Using \eqref{kosss}, there exists an $\zeta_{ij}\in[0,\epsilon]$ such that the error of approximating the term $\log(1+\epsilon\frac{a_{ij}}{b_{ij}})$ is as follows
\begin{align*}
R_2^{ij}(\epsilon)=-\frac{1}{2}\left(\frac{\frac{b_{ij}}{a_{ij}}}{1+\frac{b_{ij}}{a_{ij}}\zeta_{ij}}\right)^2\epsilon^2=-\frac{1}{2}\left( \frac{b_{ij}}{a_{ij}+b_{ij}\zeta_{ij}}\right)^2\epsilon^2.
\end{align*}
Thus, the error of approximation is as follows
\begin{align}\label{chos}
\text{Approximation\  error}&=\sum_{ij} P_j(a_{ij}+\epsilon b_{ij})R_2^{ij}(\epsilon)+\sum_{ij}P_j\frac{b_{ij}^2}{a_{ij}}\epsilon^2\nonumber
\\ &= -\sum_{ij} P_j(a_{ij}+\epsilon b_{ij})\frac{1}{2}\left( \frac{b_{ij}}{a_{ij}+b_{ij}\zeta_{ij}}\right)^2\!\!\epsilon^2\!+\!\sum_{ij}P_j\frac{b_{ij}^2}{a_{ij}}\epsilon^2
\end{align}
An upper bound on approximation error can be obtained as follows
\begin{align}\label{antala}
|\text{Approximation\  error}|&\leq |\sum_{ij} P_j(a_{ij}+\epsilon b_{ij})\frac{1}{2}\left( \frac{b_{ij}}{a_{ij}+b_{ij}\zeta_{ij}}\right)^2\epsilon^2|+|\sum_{ij}P_j\frac{b_{ij}^2}{a_{ij}}\epsilon^2|.
\end{align}
By using the definition of $\epsilon_2$ in Proposition~5 we have $\epsilon<\epsilon_2$ implies $\epsilon<\frac{\min_{ij} a_{ij}}{\max_{ij} |b_{ij}|}$, since $\min_{ij} a_{ij}=\min_{y,\Omega\in \Omega^1} M_{\Omega}^{-1}MP_Y(y)$ and $\max_{ij} |b_{ij}|<\max_{\Omega\in \Omega^1} |\sigma_{\max} (H_{\Omega})|$. By using the upper bound $\epsilon<\frac{\min_{ij} a_{ij}}{|\max_{ij} b_{ij}|}$ we can bound the second term in \eqref{antala} by $1$, since we have 
\begin{align*}
|\sum_{ij}P_j\frac{b_{ij}^2}{a_{ij}}\epsilon^2|&<|\sum_{ij}P_j \frac{b_{ij}^2}{a_{ij}}\left(\frac{\min_{ij} a_{ij}}{\max_{ij} |b_{ij}|}\right)^2|\\&<|\sum_{ij} P_j\min_{ij} a_{ij}|=|\mathcal{X}|\min_{ij} a_{ij}\\&\stackrel{(a)}{<} 1,
\end{align*}
where (a) follows from $\sum_{i} a_{ij}=1,\ \forall j\in\{1,..,|\mathcal{Y}|\}$.
\\If we use $\frac{1}{2}\epsilon_2$ as an upper bound on $\epsilon$, we have $\epsilon<\frac{1}{2}\frac{\min_{ij} a_{ij}}{\max_{ij} |b_{ij}|}$. We show that by using this upper bound the first term in \eqref{antala} can be upper bounded by $\frac{1}{2}$. We have
\begin{align*}
\frac{1}{2}|\sum_{ij} P_j(a_{ij}+\epsilon b_{ij})\left( \frac{b_{ij}}{a_{ij}+b_{ij}\zeta_{ij}}\right)^2\epsilon^2|&\stackrel{(a)}{<}\frac{1}{2}|\sum_{ij} P_j(a_{ij}+\epsilon b_{ij})\left(\frac{|b_{ij}|}{a_{ij}-\epsilon|b_{ij}|}\epsilon\right)^2|\\&\stackrel{(b)}{<}\frac{1}{2}|\sum_{ij} P_j(a_{ij}+\epsilon b_{ij})|<\frac{1}{2},
\end{align*}  
where (a) follows from $0\leq\zeta_{ij}\leq \epsilon,\ \forall i,\ \forall j,$ and (b) follows from $\frac{|b_{ij}|}{a_{ij}-\epsilon|b_{ij}|}\epsilon<1$ for all $i$ and $j$. The latter can be shown as follows
\begin{align*}
\frac{|b_{ij}|}{a_{ij}-\epsilon|b_{ij}|}\epsilon<\frac{|b_{ij}|}{a_{ij}-\frac{1}{2}\frac{\min_{ij} a_{ij}}{\max_{ij} |b_{ij}|}|b_{ij}|}\epsilon<\frac{b_{ij}}{\frac{1}{2}\min_{ij}a_{ij}}\epsilon<1.
\end{align*}
For $\epsilon<\frac{1}{2}\epsilon_2$ the term $a_{ij}-\epsilon|b_{ij}|$ is positive and there is no need of absolute value for this term. 
Thus, $\epsilon<\frac{1}{2}\epsilon_2$ implies the following upper bound
\begin{align*}
|\text{Approximation\  error}|<\frac{3}{4}.
\end{align*}
Furthermore, by following similar steps if we use the upper bound $\epsilon<\frac{1}{2}\frac{\epsilon_2}{\sqrt{|\mathcal{X}|}}$ instead of $\epsilon<\frac{1}{2}\epsilon_2$, the upper bound on error can be strengthened by
\begin{align*}
|\text{Approximation\  error}|<\frac{1}{2(2\sqrt{|\mathcal{X}|}-1)^2}+\frac{1}{4|\mathcal{X}|}.
\end{align*}
\section*{Appendix D}
\subsection*{Proofs for \emph{Privacy-utility trade-off with non-zero leakage and prioritized private data}:}
\textbf{\emph{Proof of Theorem~\ref{koonbekoon}:}}
The upper bound can be obtained using the key equation in \eqref{key}, since the total leakage $I(U;X_1,X_2)$ is bounded by $\epsilon$. The first lower bound $L_{h^{12}}^{1}(\epsilon)$ can be obtained by using EFRL stated in Lemma\ref{lemma3}. Let $U$ be produced by EFRL with $X=(X_1,X_2)$, then we have
\begin{align*}
h_{\epsilon}^{p}(P_{X_1X_2Y})&\geq I(U;Y)\\&=\epsilon+H(Y|X_1,X_2)-I(X_1,X_2;U|Y)\\&\geq\epsilon+H(Y|X_1,X_2)-H(X_1,X_2|Y).
\end{align*}
The bounds $L_{h^{p}}^{2}(\epsilon)$ and $L_{h^{p}}^{3}(\epsilon)$ can be obtained as follows. Let $\bar{U}$ be found by SFRL with $X=(X_1,X_2)$. Moreover, let $U=(\bar{U},W)$ with $W=\begin{cases}
X_2,\ \text{w.p}.\ \alpha\\
c,\ \ \text{w.p.}\ 1-\alpha
\end{cases}$, where $c$ is a constant which does not belong to $\mathcal{X}_1\cup \mathcal{X}_2 \cup \mathcal{Y}$ and $\alpha=\frac{\epsilon}{H(X_2)}$. We have
\begin{align*}
I(U;X_1,X_2)&=I(\bar{U},W;X_1,X_2)\stackrel{(a)}{=}I(W;X_1,X_2)\\&=\!H\!(X_1,\!X_2)\!-\!\alpha H(X_1|X_2)\!-\!(1\!-\!\alpha)H\!(X_1,\!X_2)\\&=\alpha H(X_2)=\epsilon,
\end{align*}
where (a) follows since $\bar{U}$ is independent of $(X_1,X_2,W)$. Furthermore, we have
\begin{align}
&I(U;X_1,X_2|Y)=I(\bar{U};X_1,X_2|Y)+I(W;X_1,X_2|Y,\bar{U})\nonumber\\&=I(\bar{U};X_1,X_2|Y)+H(X_1,X_2|Y,\bar{U})\!-\!H(X_1,X_2|Y,\bar{U},W)\nonumber\\&=I(\bar{U};X_1,X_2|Y)+\alpha H(X_1,X_2|Y,\bar{U})-\alpha H(X_1|Y,\bar{U},X_2)\nonumber \\&=I(\bar{U};X_1,X_2|Y)-\alpha H(X_1|Y,\bar{U},X_2)+\alpha\left( H(X_1,X_2|Y)-I(\bar{U};X_1,X_2|Y)\right)\nonumber\\&=\!(1\!-\!\alpha)I(\bar{U};X_1,\!X_2|Y)\!+\!\alpha H(X_1,\!X_2|Y\!)\!-\!\alpha H(X_1|Y\!,\!\bar{U}\!,\!X_2).\label{jakesh}
\end{align}
In the following we bound \eqref{jakesh} in two ways. We have
\begin{align}
\eqref{jakesh}&=\!(1\!-\!\alpha)I(\bar{U};X_1,\!X_2|Y)\!+\!\alpha H(X_2|Y)\!+\!\alpha I(X_1;\bar{U}|Y,\!X_2)\nonumber\\&=I(\bar{U};X_1,\!X_2|Y)\!+\!\alpha H(X_2|Y)\!-\!I(\bar{U};X_2|Y)\nonumber\\ &\stackrel{(a)}{\leq} \log(I(X_1,X_2;Y)+1)+4+\alpha H(X_2|Y).\label{jakesh2}
\end{align}
Furthermore,
\begin{align}
\eqref{jakesh}&\leq \!(1\!-\!\alpha)I(\bar{U};X_1,\!X_2|Y)+\alpha H(X_1,\!X_2|Y\!) \nonumber\\&\stackrel{(b)}{\leqq} \!(1\!-\!\alpha)\left(\log(I(X_1,X_2;Y)+1)+4\right)\!+\!\alpha H(X_1,\!X_2|Y\!).\label{jakesh3}
\end{align}
Inequalities (a) and (b) follow since $\bar{U}$ is produced by SFRL, so that $I(\bar{U};X_1,X_2|Y)\leq \log(I(X_1,X_2;Y)+1)+4$. Using \eqref{jakesh2}, \eqref{jakesh3} and key equation in \eqref{key} we have
\begin{align*}
h_{\epsilon}^{p}(P_{X_1X_2Y})&\geq I(U;Y)\stackrel{(c)}{\geq} \epsilon+H(Y|X_1,X_2)-\left(\log(I(X_1,X_2;Y)+1)+4+\alpha H(X_2|Y) \right)\\&=L_{h^{p}}^{2}(\epsilon),
\end{align*} 
and 
\begin{align*}
h_{\epsilon}^{p}(P_{X_1X_2Y})&\geq I(U;Y)\stackrel{(d)}{\geq}  \epsilon+H(Y|X_1,X_2)-\left((1\!-\!\alpha)\log(I(X_1,X_2;Y)+1)+4+\alpha H(X_1,X_2|Y) \right)\\&=L_{h^{p}}^{3}(\epsilon).
\end{align*} 
In steps (c) and (d) we used $H(Y|X_1,X_2,U)=0$. The latter follows by definition of $W$ and the fact that $\bar{U}$ is produced by SFRL.
\bibliographystyle{IEEEtran}
\bibliography{IEEEabrv,IZS}
\end{document}